\documentclass[journal]{IEEEtran}
\IEEEoverridecommandlockouts 			% This command is only needed if  		     									% you want to use the \thanks command

\usepackage{hyperref}
\usepackage{graphics} % for pdf, bitmapped graphics files
\usepackage{epsfig} % for postscript graphics files
\usepackage{times} % assumes new font selection scheme installed
\usepackage[tbtags]{amsmath} % assumes amsmath package installed
\usepackage{amssymb}  % assumes amsmath package installed
\usepackage{amsthm}
\usepackage{mathtools}
\usepackage{cite}
\usepackage{algorithm}
\usepackage{algorithmic}
\usepackage{xcolor}
\usepackage{tikz}
\usepackage{graphicx}
\newtheorem{theorem}{\textbf{Theorem}}

\newtheorem{lemma}{\textbf{Lemma}}
\newtheorem{corollary}{\textbf{Corollary}}
\newtheorem{definition}{Definition}

\newtheorem{remark}{\textbf{Remark}}

\usepackage{arydshln}
\usepackage{subcaption}
\usetikzlibrary{shapes}
\usepackage{float}
\usetikzlibrary{calc}
\usepackage{enumerate}

\newcommand{\cfbox}[2]{%
    \colorlet{currentcolor}{.}%
    {\color{#1}%
    \fbox{\color{currentcolor}#2}}%
}

\newcommand{\real}{\mathbb{R}}%Real 
\newcommand{\graph}{\mathcal{G}}%Calligraphic G, if needed
\newcommand{\mean}{\mathbb{E}}
\newcommand{\pa}{\mathcal{P}}
\newcommand{\ch}{\mathcal{C}}
\newcommand\indep{\protect\mathpalette{\protect\independenT}{\perp}}
\def\independenT#1#2{\mathrel{\rlap{$#1#2$}\mkern3mu{#1#2}}}

\newcommand{\Graph}{\mathcal{G}}

\title{\LARGE \bf
  Network Structure Identification from Corrupt Data Streams
}

\author{Venkat Ram Subramanian, Andrew Lamperski, and Murti V. Salapaka
\thanks{
The authors are with the Department of Electrical and Computer
   Engineering, University of Minnesota, Minneapolis, MN 55455, USA.{\tt\small subra148@umn.edu, alampers@umn.edu,  murtis@umn.edu}}
   \thanks{Work supported in part by NSF CMMI 1727096. }
}

\begin{document}
\maketitle
%%%%%%%%%%%%%%%%%%%%%%%%%%%%%%%%%%%%%%%%%%%%%%%%%%%%%%%%%%%%%%%%%%%%%%%%%%%%%%
\begin{abstract}%will be modified later
Complex networked systems can be modeled as graphs with nodes
representing the agents and links describing the dynamic coupling
between them.
Previous work on network identification has shown that the network
structure of linear time-invariant (LTI) systems can be reconstructed
from the joint power spectrum of the data streams.
These results assumed that data is perfectly measured.
However, real-world data is
subject to many corruptions, such as inaccurate time-stamps, noise,
and data loss.
% In this article, we first present a general
% form of perturbation models as stochastic linear dynamical systems
% with the ideal data-streams as \textit{input}.
We show that
identifying the structure of linear time-invariant systems using
corrupt measurements results in the inference of erroneous links. We provide an exact characterization and prove that such erroneous links are restricted to
the neighborhood of the perturbed node. We extend the analysis of LTI
systems to the case of Markov
random fields with corrupt measurements. We show that data corruption
in Markov random fields results in spurious probabilistic
relationships in precisely the locations where spurious links arise in
LTI systems.
% This result indicates that algorithms that use
% probabilistic relationships to infer network structure are likely to
% include erroneous links when data is corrupted.  
\end{abstract}
%%%%%%%%%%%%%%%%%%%%%%%%%%%%%%%%%%%%%%%%%%%%%%%%%%%%%%%%%%%%%%%%%%%%%%%%%%%%%%
\section{Introduction}
%\subsection{Motivation}
Identification of network interaction structures is important for
several domains such as climate science \cite{fan2017elnino},
epidemiology \cite{kaufman2017methods},
neuroscience\cite{bassett2017network}, metabolic pathways
\cite{julius2011genetic}, quantitative finance
\cite{giudici2016graphical}\cite{kenett2015network}, the
internet-of-things \cite{perera2014context}\cite{guinard2010interacting} and video
streaming \cite{WanTan2003}.
% The need for frameworks for analyzing and
% designing 
% networks of dynamical systems has received renewed emphasis fueled by
% new technologies and paradigms. For example, in the internet-of-things (see \cite{iotVisonChallenges2010, datla2012wireless}), data
% collected from ubiquitously sensorized devices and/or 
% from many sensors of a single large system are
% processed to glean important insights into the network of interacting
% components. Here, data-steams from various sensors can be dynamically
% related, where the interdependency can be caused by the interaction
% physics of system components. Whether a system is comprised of a  network of interacting agents or modeled as such, or both, often, the interactions between agents are dynamic and include feedback loops. In these systems, identification of influence
% pathways and determining the topology of the underlying interaction network is of significant interest.
In scenarios such as the power grid\cite{deka2018structure} and
financial markets it is impractical, impossible or impermissible to
externally influence the system. Here network structure identification
must be achieved via passive means. The passive identification of a
network of dynamically related agents is becoming more viable with
sensors and measurements becoming inexpensive coupled with the ease
and capability of communicating information.  

Often, the measurements in such large systems are subjected to
effects of noise \cite{stankovic2018distributed}, asynchronous sensor
clocks \cite{cho2014survey} and packet drops
\cite{leong17sensor}. When dealing with problems of identifying
structural and functional connectivity of a large network, there is a
pressing need to rigorously study such uncertainties and address
detrimental effects of corrupt data-streams on network
reconstruction. 
Such analysis can delineate
the effects of corrupted nodes on the quality of the network
reconstruction and suggest placement of high-fidelity sensors at
critical nodes.
% Moreover how uncertainty at  particular node in the network
% leads to spuriousness in identification of links globally is important
% to be understood.
% Insights on how the uncertainty propagates will allow
% system architects to bolster neighborhoods that have large systemic
% influences with better data measurement and communication
% resources.
\subsection{Related Work}
Network identification for linear systems has been extensively
studied. Below, we will give an overview of several research
themes in linear system network identification. However, the majority of works assume that the
measurements are perfect. 

Identifiability conditions for determining the transfer
functions are provided in \cite{weerts2018identifiability}. It is
shown that a network is identifiable if every node signal
is excited by either an external input or a noise signal that is
uncorrelated with the input/noise signals on the other nodes.
The effects of data corruption are not studied in this work.
% However,
% effects of data-corruption on the determining the boolean structure or
% characterizing its effects the transfer function estimates are not
% investigated. 

For partially observed states, authors in \cite{hendrickx2018id}
provide necessary and sufficient conditions for \emph{generic} identifiability of all or a subset of the
transfer functions in the network. Similarly, the notion of \emph{global} 
idenitifiability has been studied in \cite{van2019necessary}. 
However, in both the articles, the topology of the
network is assumed to be known a priori.
Moreover, data measurements are assumed to be perfect. 

The problem of learning polytree structures has been studied in \cite{etesami2016learning} and \cite{sepehr2019blind}. 
The authors provide guarantees of a consistent reconstruction. However, the class of network structures was restricted to trees and
the data measurements are assumed to be ideal. In this article, we make no such assumptions on network structures and we study the problem
 when time-series data measurements are imperfect.

Authors in \cite{MatSal12} leveraged  multivariate Wiener filters to reconstruct the undirected topology of the
generative network model. With  assumptions of perfect measurements,
and linear time invariant interactions, it is established that the
multivariate Wiener filter can recover the \emph{moral graph}. In other
words, for each node, its parents, co-parents and children are detected. 
% Here, spurious influences are detected; however, they remain local and within one hop of a true interaction. t is imminent and
% desirable to establish conditions under which the combination of
% spurious links from the method of \cite{MatSal12} and the effects of
% data corruption do not compound so that negative effects remain
% localized.
%
%Moreover, assuming that the interaction dynamics are {\it strictly causal} and using multivariate estimation based on a Granger filter, it was shown that the interaction structure can be accurately recovered with directions, and without any spurious links.
%As alluded to earlier, results assume perfect measurements and LTI dynamics.

For a network of interacting agents with nonlinear dynamics
 and  strictly causal interactions, the authors in
\cite{quinn15DIG} proposed the use of directed information to
determine the directed structure of the network. Here too,
it is assumed that the data-streams are ideal with no distortions. 

The authors in \cite{yuan2011robust},\cite{chetty2013robust} use
dynamical structure functions (DSF) for network reconstruction
\cite{goncalves2008necessary} and consider measurement noise and
non-linearities in the network dynamics. The proposed method first
finds optimal DSF for all possible Boolean structures and then adopt a
model selection procedure to determine the best estimate. The authors
concluded that the presence of noise and non-linearities can even lead
to spuriously inferring fully connected network structures. Also, the
authors concluded that the performance of their algorithms degrades as
noise, network size and non-linearities increase. However, a precise
characterization of such spurious inferences in structure was not
provided.

%preface it: extension of earlier works
%In \cite{SLS17network}, the analysis is restricted to
%LTI systems. Moreover, the directions of the spurious interactions are
%not determined. In \cite{SLS18inferring} underlying dynamics generating the data is allowed to be non-linear
%and admits feedback loops. Here, the interactions were assumed to be
%strictly causal. 
%The authors determined directions of spurious links that can arise when inferring network
%structure from corrupt data-streams. However, the analysis is limited to generative models that allow only \emph{faithful} dynamics.  
\subsection{Our Contribution}
In this article, our problem of interest is to determine the Boolean
structure of a network, using passive means from corrupt
data-streams and characterize the spurious
links that can appear due to data-corruption.

In order to rigorously model data corruption, we present a general
class of signal disturbance models based on randomized state-space
systems. This class of disturbances subsumes many uncertainties that are prevalent in applications. We provide a detailed description on how the corruption model affects the second order statistics of the data-streams.

Next, we present the results for inferring the network topology for LTI systems from corrupt data-streams. Specifically, we identify a set of edges in the network in which
spurious links could potentially appear. The
results can be utilized to understand what part of the
reconstruction can be trusted and to allocate sensor resources in
order to
minimize the effects of data corruption.
%This accentuates the pertinence of the methods for wider class of networks that has feedback loops and allows instantaneous and/or non-causal dynamic relations.
 
Finally, we extend our analysis and provide connections with more
general graphical
models. We prove that there can be spurious edges inferred during
structure identification of undirected Markov random fields from
corrupt data. The results characterizing the location of the spurious
links are found to be identical to those obtained in LTI
systems.

This paper is an extension of our earlier work \cite{SLS17network}
wherein preliminary results characterizing the spurious links were
presented.  However, a rigorous description on the perturbation models
was not provided, and the work did not cover Markov random
fields.

\subsection{Paper Organization} 
We start by reviewing earlier work on LTI network identification using
power spectra in Section~\ref{sec:prelim}.  In Section~\ref{sec:NW pert},
we describe our data corruption models. In Section ~\ref{sec:LTI}, we
characterize the spurious links due to data corruption for LTI systems. Section ~\ref{sec:mrf} discusses
the effects of data-corruption in inferring the undirected structure
of a Markov random field. Simulation results are provided in
Section~\ref{sec:results}. Finally, a conclusion is provided in Section
~\ref{sec:conclude}. 

\subsection{Notation}
\noindent $Y$ denotes a vector with $y_i$ being $i^{th}$ element of $Y.$ \\
$z_i [\cdot ]$ denotes a sequence and $z_{i,t}$ denotes $z_i[t]$.\\
$\parallel \cdot \parallel$ denotes standard Eucledian norm for vectors.\\
$P_X$ represents the probability density function of a random variable $X$.\\
$X\indep Y$ denotes that the random variables $X$ and $Y$ are independent.\\
$i\to j$ indicates an arc or edge from node $i$ to node $j$ in a directed graph.\\  
$i - j$ denotes an undirected edge between nodes $i$, $j$ in an undirected graph.\\
If $M(z)$ is a transfer function matrix, then $M(z)^* = M(z^{-1})^T$ is the conjugate transpose. \\
$\mathbb{E}[\cdot ]$ denotes expectation operator.\\
$R_{XY}(k):=\mathbb{E}[X[n+k]Y[n]]$ is the cross-correlation function of jointly wide-sense stationary(WSS) processes $X$ and $Y$. If $Y=X$ then $R_{XX}(k)$ is called the auto-correlation.\\
$\Phi _{XY}(z):=\mathcal{Z}(R_{XY}(k))$ represents the cross-power spectral density while $\Phi _{XX}(z):=\mathcal{Z}(R_{XX}(k))$ denotes the power spectral density(PSD) where
$\mathcal{Z}(\cdot )$ is the Z-transform operator.\\
$b_i$ represents the $i^{th}$ element of the canonical basis of
$\real ^n$.

%%%%%%%%%%%%%%%%%%%%%%%%%%%%%%%%%%%%%%%%%%%%%%%%%%%%%%%%%%%%%%%%%%%%%%%% 
\section{Background on LTI Network Identification}\label{sec:prelim}
%%%%%%%%%%%%%%%%%%%%%%%%%%%%%%%%%%%%%%%%%%%%%%%%%%%%%%%%%%%%%%%%%%%%%%%%%%%%%%

This section reviews earlier results on network identification from
ideal data streams. See
\cite{MatSal12}. Required graph theoretic notions are described in
Subsection~\ref{sec:graphPrelim}. The formal model of networked LTI
systems is presented in Subsection~\ref{sec:DIM}. Then, a result on
network identification via power spectra is given in Subsection~\ref{sec:SIideal}.
In later sections, we
will analyze these results in the case that data has been corrupted.

\subsection{Graph Theoretic Preliminaries}
\label{sec:graphPrelim}
We will review some terminology from graph theory needed to describe the
background results on LTI identification. For reference, see\cite{koller2009prob}.

\begin{definition}[Directed and Undirected Graphs]
\label{def:Graphs}
A \emph{directed graph} $G$ is a pair $(V,A)$ where $V $ is a
set of vertices or nodes and $A$ is a set of edges given by ordered pairs $(i,j)$
where $i,j\in V$. If $(i,j) \in A$, then we say that there is an edge
from $i$ to $j$.  $(V,A)$ forms an \emph{undirected graph} if $V$ is a
set of nodes or vertices and $A$ is a set of the
un-ordered pairs
$\{i,j\}$. 
\end{definition}
We also denote an undirected edge as $i-j$. 

\begin{definition}[Children and Parents]
\label{def:ChP}
Given a directed graph ${G}=(V,A)$ and a node $j\in V$, the children of $j$ are defined as $\ch (j):=\left\lbrace i|j\to i \in A\right\rbrace $ and the parents of $j$ as $\pa (j):=\left\lbrace i|i\to j \in A\right\rbrace $.
\end{definition}
\begin{definition}[Kins]
\label{def:Kins}
Given a directed graph $G=(V,A)$ and a node $j\in V$, kins of $j$ are defined as $\mathcal{K}_j:=\left\lbrace i|i\neq j \text{ and } i \in \mathcal{C}(j) \cup \mathcal{P}(j) \cup \mathcal{P}(\mathcal{C}(j))\right\rbrace $.
Kins are formed by parents, children and spouses. A spouse of a node is another node where both nodes have at-least one common child.
\end{definition}
\begin{definition}[Moral-Graph]
\label{def:KinGraph}
Given a directed graph $G=(V,A)$, its moral-graph is the undirected graph $G^M=(V,A^M)$ where $A^M:=\left\lbrace
  \{i,j\}|
  j\in V,i\in \mathcal{K}_j\right\rbrace .$ 
\end{definition}

Fig. \ref{fig:intro Graph} provides an example of a directed graph and its moral graph. 

%%%%%%%%%%%%%%%%%%%%%%%%%%%%%%%%%%%%%%%%%%%%%%%%%%%%%%%%%%%%%%%%%%%%%%%%%%%%%%
\begin{figure}
\centering
  \begin{minipage}{0.35\linewidth}
    \centering
\begin{tikzpicture}[scale=0.55]
 \tikzstyle{vertex}=[circle,fill=none,minimum size=10pt,inner sep=0pt,thick,draw]
 \tikzstyle{pvertex}=[star,star points=10,fill=white,minimum size=10pt,inner sep=0pt,thick,draw]
  \node[vertex] (n1) {$1$};

  \node[vertex] (n2) at ($(n1) - (2em,3em)$) {$2$};%

        \node[vertex] (n3) at ($(n1) - (-2em,3em)$) {$3$};%

    \node[vertex] at ($(n1)-(0,6em)$) (n4) {$4$};%
    %\node[vertex,below of=n4] (n5) {$5$};
	\node[vertex] at ($(n4)-(0,4em)$) (n5){$5$};
  \draw[->,thick] (n1)--(n2);
  \draw[->,thick] (n1)--(n3);
  \draw[->,thick] (n2)--(n4);
  \draw[->,thick] (n3)--(n4);
  \draw[->,thick] (n4)--(n5);
\end{tikzpicture}
 \subcaption{\label{fig:DG}}
 \end{minipage}\hfill
 \begin{minipage}{0.35\linewidth}
    \centering
\begin{tikzpicture}[scale=0.55]
 \tikzstyle{vertex}=[circle,fill=none,minimum size=10pt,inner sep=0pt,thick,draw]
 \tikzstyle{pvertex}=[star,star points=10,fill=white,minimum size=10pt,inner sep=0pt,thick,draw]
  \node[vertex] (n1) {$1$};

  \node[vertex] (n2) at ($(n1) - (2em,3em)$) {$2$};%

        \node[vertex] (n3) at ($(n1) - (-2em,3em)$) {$3$};%

    \node[vertex] at ($(n1)-(0,6em)$) (n4) {$4$};%
    	\node[vertex] at ($(n4)-(0,4em)$) (n5){$5$};

  \draw[-,thick] (n1)--(n2);
  \draw[-,thick] (n1)--(n3);
  \draw[-,thick] (n2)--(n4);
  \draw[-,thick] (n3)--(n4);
  \draw[-,thick] (n4)--(n5);
    \draw[-,thick] (n2)--(n3);
      
\end{tikzpicture}
 \subcaption{\label{fig:Kin Graph}}
 \end{minipage}
 \caption{
   \label{fig:intro Graph} 
   \ref{fig:DG} Directed Graph
    and \ref{fig:Kin Graph} its moral Graph.}
\end{figure}
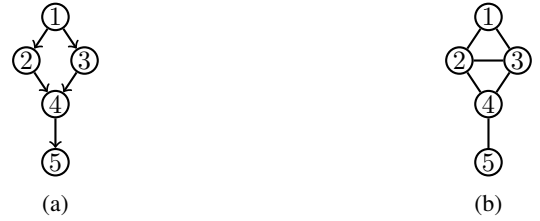
%%%%%%%%%%%%%%%%%%%%%%%%%%%%%%%%%%%%%%%%%%%%%%%%%%%%%%%%%%%%%%%%%%%%%%%%%%%%%%%%%%%%%%%%%
%
%
%%%%%%%%%%%%%%%%%%%%%%%%%%%%%%%%%%%%%%%%%%%%%%%%%%%%%%%%%%%%%%%%%%%%%%%%
\subsection{Dynamic Influence Model for LTI systems}\label{sec:DIM}
Here the \textit{generative model} that is assumed to generate the
measured data is described. Consider $N$ agents that interact over a
network. For each agent $i$, we associate an observable discrete time sequence $y_i[\cdot]$
and a hidden noise sequence $e_i[\cdot].$ The process $e_i[\cdot]$ is considered
innate to agent $i$ and thus $e_i$ is independent of $e_j$ if $i\not=
j.$ We assume $e_i$ and $y_i$ to be jointly wide-sense stationary stochastic
processes. In particular, we assume they are bounded in a mean-square
sense: $\mean [\parallel  y_i[t]\parallel  ^2] <\infty $ and $\mean
[\parallel e_i[t] \parallel ^2] < \infty$.

Let $Y$ denote the set
of all random process $\{y_1,\ldots, y_N\}$ with a parent set
$\mathcal{P}' (i)$ defined for $i=1,\ldots,N.$ The parent set
$\mathcal{P}'(i)$ associated with agent $i$ does not include $i$.
 The process $y_i$ depends dynamically on the processes of its parents, $y_j$
with $j\in \mathcal{P}'(i)$ through an LTI filter whose impulse response is given by $\graph_{ij}$. Therefore, dynamics of node $i$ takes the form:
\begin{equation}\label{eq:idyntime}
y_i[t]=\sum^{N}_{j\in \mathcal{P}'(i)}(\graph _{ij} * y_j)[t] + e_i[t] \ \ \mbox{for
}i=1,\ldots, N. 
\end{equation}
where $*$ denotes convolution operation. Performing a $Z$-transform on both
sides gives
\begin{equation}\label{eq:iDyn}
y_i(z) = \sum^{N}_{j\in \mathcal{P}'(i)}\graph _{ij}(z)y_j(z) + e_i(z) \ \ \mbox{for
}i=1,\ldots, N. 
\end{equation}

For compact notation, we will often drop the $z$ arguments. Let $y=(y_1,y_2,\dots,y_N)^T$ and $e=(e_1,e_2,\dots,e_N)^T$. Then \eqref{eq:iDyn} is equivalent to 
\begin{equation}
\label{eq:ldgCompact}
y = \graph (z) y + e.
\end{equation}
The diagonal entries $\graph _{ii}(z)$ are considered to be zero.
We refer to \eqref{eq:ldgCompact} as the Dynamic Influence Model (DIM). Here, $\graph $ is termed as the DIM {\it generative connectivity}
matrix. The
DIM will be denoted by $(\graph ,e)$.
% and the underlying graph, $\mathcal{G}$ is called the \emph{generative directed graph} of $(\bar{g},e)$.
\begin{remark}\label{rem:corrNoise}
The process noise in \eqref{eq:idyntime} can be correlated across
time. In that case, $e_i$ is assumed to be represented as the convolution of
white noise with a stable LTI filter.
\end{remark}
\begin{remark}\label{rem:diagonalG}
The diagonal entries, $\Graph _{ii}(z)$ are considered to be zero only for simplification purposes to remove self-dependence in the dynamics. As will be seen later in sub-section ~\ref{sec:SIideal}, this enables us to consider Wiener filter projection of signal $y_i$ on all signals except $y_i$. Moreover, we can model the self-dependence and include it in the DIM through the process noise sequence by convolving a zero mean white noise with $\Graph _{ii}(z)$.
\end{remark}
We illustrate the notation by an example. Consider a network of five agents whose node dynamics are given by,
\begin{equation}\label{eq:eggenmod1}
\begin{aligned} 
y_1&= e_1\\
y_2&= \graph _{21}(z)y_1+e_2\\
y_3&= \graph _{31}(z)y_1+e_3\\
y_4&= \graph _{42}(z)y_2+\graph_{43}(z)y_3+e_4\\
y_5&= \graph _{54}(z)y_4+e_5
\end{aligned}
\end{equation} 
with $\graph=\left[\begin{array}{ccccc}
0 & 0 & 0 & 0&0\\
\graph _{21} & 0 & 0 & 0&0\\
\graph _{31}& 0 & 0 & 0 & 0\\
0 & \graph _{42}&\graph _{43}&0&0\\
0&0&0& \graph_{54} &0\\
\end{array}\right]
$. 
\begin{definition}[Generative Graph]\label{def:gengraph}
The structural description of  \eqref{eq:ldgCompact} induces a {\it generative graph} $G=(V,A)$ formed by identifying each vertex $v_i$ in $V$ with random process $y_i$ and the set of directed links, $A,$ obtained by introducing a directed link from every element in the parent set $\mathcal{P}'(i)$ of agent $i$ to $i.$ 
\end{definition}
Note that we do not show $i\to i$ in the generative graph and neither do we show the processes $e_i$. The generative graph associated with the examples described in ~\eqref{eq:eggenmod1} is given by Fig. \ref{fig:intro Graph} (a). 
%%%%%%%%%%%%%%%%%%%%%%%%%%%%%%%%%%%%%%%%%%%%%%%%%%%%%%%%%%%%%%%%%%%%%%%%%
\subsection{Identification from Ideal Measurements}\label{sec:SIideal}
The following results are obtained from \cite{MatSal12} where the authors have leveraged Wiener filters for determining generative graphs of a DIM. 
\begin{theorem}
\label{thm:KinshipInf}
{\it
Consider a DIM $(\graph,e)$ consisting of N nodes with generative graph $G$. Let the output of the DIM be given by $y=(y_1,\dots,y_N)^T$. Suppose that $S_j$ is the span of all random variables $y_k[t],\  t=\ldots -2,-1,0,1,2\ldots$ excluding $y_j$. Define the estimate $\hat{y}_j$ of the time-series $y_j$  via the optimization problem of
\begin{align*}\label{eq:WF}
\underset{\hat{y}_j \in S_j}{\text{min}}{\mathbb{E}\left[{(y_j-\hat{y}_j)}^T{(y_j-\hat{y}_j)}\right]}.
\end{align*}
Then a unique optimal solution to the above exists and is given by
\begin{equation}\label{eq:Non causal WF}
\hat{y}_j=\sum _{i\neq j}\mathbf{W}_{ji}(z)y_i
\end{equation}
where $\mathbf{W}_{ji}(z)\neq 0$ implies $y_i\in \mathcal{K}_{y_j}$ (equivalently $y_j\in \mathcal{K}_{y_i}$); that is $i$ is a kin of $j$.
}
\end{theorem}

The solution in (\ref{eq:Non causal WF}) is the Wiener Filter solution which is given by  $\Phi _{y_jy_{\bar {j}}} \Phi ^{-1} _{y_{\bar {j}}{y_{\bar {j}}}}$ where $y_{\bar{j}}$ denotes the vector of all processes excluding $y_j$ and $\Phi$ denotes the power spectral density. 
Thus, Theorem~\ref{thm:KinshipInf} implies that we can reconstruct the
moral graph of a DIM by analyzing the joint power spectral density of
the measurements.
The following corollary gives a useful characterization of the
inferred kin relationships in terms of the sparsity pattern of
$\Phi_{yy}^{-1}$. 

% \textit{Proof:}Refer \cite{MatSal12} for proof.
\begin{corollary}\label{cor:informationSpectra}
{\it
  Under the assumptions of Theorem \ref{thm:KinshipInf}, let $\Phi_{yy}$ be the power spectral density matrix  of the vector process $y$. Then the $(j,i)$ entry of $\Phi_{yy}^{-1}$ is non zero implies that $i$ is a kin of $j$.
  }
\end{corollary}
\begin{remark}\label{rem:otherdirectionMatSal}$\Phi_{yy}^{-1}(i,j)$ is described by $(i,j)$ entry of $(I-\Graph (z))^*\Phi ^{-1}_{e}(I-\Graph (z))$. Specifically, $\Phi_{yy}^{-1}(i,j)=-\Graph _{ij}\phi _{e_i}^{-1}-\Graph _{ji}^*\phi _{e_j}^{-1}+\sum _{k}\Graph _{ki}^*\Graph _{kj}\phi _{e_k}^{-1}$  where $k \in \ch (i)\cap \ch (j)$. For $i$ and $j$ being kins but $\Phi_{yy}^{-1}(i,j)$ to be zero, the transfer functions in $\Graph $ must be  belong to a set of measure zero on space of system parameters. For example,  system dynamics with transfer functions being zero or a static system with all noise sequences being identical. Therefore, except for these restrictive cases, the results in Theorem ~\ref{thm:KinshipInf} and Corollary \ref{cor:informationSpectra} are both necessary and sufficient. See \cite{MatSal12} for more details. 
\end{remark}
% \textit{Proof:}Refer \cite{MatSal12} for proof.\\
\section{Uncertainty Description }\label{sec:NW pert}
Subsection~\ref{sec:SIideal} describes a
methodology from \cite{MatSal12} for guaranteed kinship reconstruction based on Wiener
filtering. However, the results assume that the signals, $y_i$, are
measured perfectly. This paper aims to explain what would happen if we
attempted to apply the reconstruction method to data that has been
corrupted. We will see that extra links appear in the reconstruction,
and characterize the pattern of spurious links. While the analysis of
the next two sections 
focuses on LTI identification,
the results on Markov random fields in
Section~\ref{sec:mrf} indicate that the emergence and pattern of
spurious links are general properties of network reconstruction from
corrupted data.  

Subsection~\ref{sec:RSS} presents the general class of data
corruption models studied for LTI systems. The modeling framework is sufficiently general to capture
a variety of practically relevant perturbations, such as delays and
packet loss. However, we will see that all of the corruption models have similar
effects on the observed power spectra. Specific examples of
perturbation models are described in Subsection~\ref{sec:corruptionExamples}.

\subsection{Random State Space Models}
\label{sec:RSS}
%Both the packet-dropping link and the random delay model can be cast
This subsection presents the general class of perturbation models. 
Consider $i^{th}$ node in a network and let it's associated unperturbed time-series be $y_i$. The corrupt data-stream $u_i$ associated with $i$ is considered to follow the stochastic linear system described below:
\begin{subequations}
\label{eq:stateSpace}
\begin{align}
  x_i[{t+1}] &= A_i[t] x_i[t] + B_i[t] y_i[t] + w_i[t]\\
  u_i[t] &= C_i[t] x_i[t] + D_i[t] y_i[t] + v_i[t],
\end{align}
\end{subequations}
where $x_i$ denotes hidden states in the stochastic linear system that describes the corruption.
Here, the matrices, $M_i[t] = \begin{bmatrix}
  A_i[t] & B_i[t] \\
  C_i[t] & D_i[t]
\end{bmatrix}$ are independent, identically distributed (IID) and
independent of $y_i[t]$. The terms $w_i[t]$ and $v_i[t]$ are
zero-mean IID noise terms which are independent of $M_i[\cdot]$ and
$y_i[\cdot]$ and have  covariance:
\begin{equation}
  \label{eq:noiseCov}
  \mean\left[
    \begin{bmatrix}
      w_i[t] \\
      v_i[t]
    \end{bmatrix}
    \begin{bmatrix}
      w_i[t] \\
      v_i[t]
    \end{bmatrix}^\top
  \right] = \begin{bmatrix}
    W & S \\
    S^\top & V
    \end{bmatrix}.
  \end{equation}
  For distinct perturbed nodes, $i\ne j$, we assume that $M_i[]$,
  $w_i[]$, and $v_i[]$ are
  independent of $M_j[]$, $w_j[]$, and $v_j[]$.

  Denote the means of the state space matrices by $\bar A_i = \mean[A_i[t]]$, $\bar B_i = \mean[B_i[t]]$,
  $\bar C_i = \mean[C_i[t]]$, and $\bar D_i = \mean[D_i[t]]$.

Let $h_i$ be the impulse response of the system defined by
$\bar{A}_i,\bar{B}_i,\bar{C}_i,\bar{D}_i$:
\begin{equation}
\label{eq:impulse}
  h_i(k) = 
\left[
      \begin{array}{c|c}
        \bar A_i & \bar B_i
        \\
        \hline
        \vspace{-10px}
        \\
                \bar{C}_i & \bar D_i
      \end{array}
      \right](k)
    \end{equation}

    Note that $\bar u_i[t]=\mean[u_i[t]| y_i]=(h_i
    \star y_i)[t]$. 
\begin{theorem}\label{thm:perturbStats}
Assume that $M_i[t]$ has bounded
second moments and for all positive definite matrices $Q$, the following
generalized Lyapunov equation has a unique positive definite solution, $P$:
\begin{equation}
  \label{eq:genLyap}
  P = \mean[A_i[t] P A_i[t]^\top] + Q .
\end{equation}

Define $\Delta
u_i[t] := u_i[t]-\bar u_i[t]$. Then, the signals $u_i$ will be wide sense-stationary with cross-spectra and power spectra of the form:
\begin{subequations}
\label{eq:spectrumConditions}
\begin{align}
  \label{eq:uSpectrum}
\Phi_{u_i u_i}(z) &= H_i(z) \Phi_{y_i y_i}(z)H_i(z^{-1}) + \theta _i(z)\\  
  \label{eq:crossSpectrum}
  \Phi_{u_i y_i}(z) &= H_i(z) \Phi_{y_iy_i}(z)
\end{align}
\end{subequations} where, $H_i(z)=\mathcal{Z}(h_i)$ and $\theta _i(z)=\mathcal{Z}\left(R_{\Delta u_i\Delta u_i}[k]\right)$.
\end{theorem}

The proof is given in Appendix~\ref{app:perturbationProof}. 

\subsection{Data Corruption Examples}
\label{sec:corruptionExamples}
We will highlight a few corruptions that are practically relevant to
exemplify the above model description. More complex perturbations can
be obtained by composing these models. 

\subsubsection{Random Delays}
Randomized delays can be modeled by
\begin{equation}
\label{eq:randDelayMdl}
u_i[t] = y_i[t-d[t]]
\end{equation}
where $d[t]$ is a random variable. For example, if $d[t] \in
\{1,2,3\}$, then randomized delay model can be represented in
state-space form with no additive noise terms and state space matrices
given by:
\begin{equation*}
  \left[
  \begin{array}{c:c}
    A_i[t] & B_i[t] \\
    \hdashline
    C_i[t] & D_i[t]
  \end{array}
  \right]
    =
    \left[
      \begin{array}{c:c}
           \begin{bmatrix}
             0 & 0 & 0 \\
             1 & 0 & 0 \\
             0 & 1 & 0
           \end{bmatrix}
               &
\begin{bmatrix}
  1 \\ 0 \\ 0
\end{bmatrix} \\
              \hdashline
              b_{d[t]}^\top & 0
  \end{array}
                     \right],
                   \end{equation*}
where $b_1$, $b_2$, and $b_3$ are the standard basis vectors of
$\mathbb{R}^3$.

Say that $d[t]=j$ with probability $p_j$, for $j=1,2,3$. Then
\begin{equation}
  H_i(z) = p_1 z^{-1} + p_2 z^{-2} + p_3 z^{-3}.
\end{equation}
Let $p = \begin{bmatrix} p_1 & p_2 & p_3\end{bmatrix}$. The formal description to compute the 
expression for $\theta_i(z)$ is discussed in Lemma ~\ref{lem:deltaU} contained in the Appendix section. Using
Lemma~\ref{lem:deltaU} we have that $R_{\Delta u_i \Delta u_i}[t]=0$
for $t\ne 0$ and $ R_{\Delta u_i \Delta u_i}[0]$ is given by
\begin{equation}
R_{y_iy_i}[0] - p^\top \begin{bmatrix}
R_{y_iy_i}[0] & R_{y_iy_i}[1] & R_{y_iy_i}[2]\\
R_{y_iy_i}[-1] & R_{y_iy_i}[0] & R_{y_iy_i}[1]\\
R_{y_iy_i}[-2] & R_{y_iy_i}[-1] & R_{y_iy_i}[0]
\end{bmatrix}
p.
\end{equation}

\subsubsection{Measurement Noise}
White measurement noise can be represented in the form of
\eqref{eq:stateSpace} by setting $C_i[t]=0$, $D_i[t]=1$:
\begin{equation}\label{eq:measNoise}
  u_i[t] = y_i[t] + v_i[t].
\end{equation}
Colored
measurement noise with rational spectrum arises when $B_i[t]=0$, $D_i[t]=1$, and
the matrices $A_i[t]$ and $C_i[t]$ are constant. More generally, the
result of causally filtering the signal and then adding noise can be
modeled by taking all of the matrices in \eqref{eq:stateSpace} to be
constant. 

For the corruption model described in \eqref{eq:measNoise}, the perturbation transfer functions are given by:
\begin{align*}
  H_i(z) &= 1 \\
  \theta_i(z) &= \Phi_{v_i v_i}(z).
\end{align*}

\subsubsection{Adversarial Disinformation}
This is an example of data-corruption that is pertinent to
cyber-security. Here, the true data stream $y_i$ is completely
concealed and a new false data stream $v_i$ is introduced. This is an
extreme case of \eqref{eq:stateSpace} in which $C_i[t]$ and $D_i[t]$ are zero:
%This can be observed as a an extreme case of \eqref{eq:noiseNC} when $g'=0$.
\begin{equation}
\label{eq:advDis}
u_i[t]=v_i[t]
\end{equation}

\subsubsection{Packet Drops}
Here the data stream suffers from randomly dropping measurement packets. The corrupted data stream $u_i$ is obtained from $y_i$ as follows:
\begin{equation}
 \label{eq:packetDrop}
u_i[t]=\begin{cases}
y_i[t], & \textrm{ with probability } p_i\\
u_i[t-1], & \textrm{ with probability } (1-p_i)
\end{cases}
\end{equation}
Packet drops can be modeled in the form of \eqref{eq:stateSpace} with
no noise and matrices given by:
\begin{equation}
  \label{eq:packetSS}
  \begin{bmatrix}
    A_i[t] & B_i[t] \\
    C_i[t] & D_i[t]
  \end{bmatrix}
  = \begin{cases}
    \begin{bmatrix}
      0 & 1 \\
      0 & 1
    \end{bmatrix} & \textrm{ with probability } p_i \\
    \vspace{-10px}
    \\
    \begin{bmatrix}
      1 & 0 \\
      1 & 0
    \end{bmatrix} & \textrm{ with probability } 1-p_i.
  \end{cases}
\end{equation}
The generalized Lyapunov equation becomes:
\begin{equation}
  P = p_i P \cdot 0 + (1-p_i) P \cdot 1 + Q
\end{equation}
which has the solution $P = Q/p_i$. Thus, the conditions for
Theorem~\ref{thm:perturbStats} hold, and so $u_i$ is wide-sense
stationary. 
In this case
\begin{equation}
  \begin{bmatrix}
    \bar A_i & \bar B_i \\
    \bar C_i & \bar D_i
  \end{bmatrix} = \begin{bmatrix}
    1-p_i & p_i \\
    1-p_i & p_i
  \end{bmatrix}
\end{equation}
so that $H_i(z) = \frac{p_i(1-p_i)}{z-(1-p_i)} + p_i =
\frac{p_i}{1-z^{-1}(1-p_i)}$.

The formal description to compute the 
expression for $\theta_i(z)$ is discussed in Lemma ~\ref{lem:deltaU} contained in the Appendix section.
The application of methods described in the Appendix to derive an expression for
$\theta_i(z)$ is cumbersome. However, $\theta_i(z)$ can be calculated
directly. Indeed, direct calculation shows that
\begin{equation}
  \label{eq:packetMean}
  (h_i \star R_{yy} \star h_i^*)[t] =
  \sum_{j=-\infty}^{|t|}\sum_{k=j}^{\infty} p_i^2 (1-p_i)^{|t|+k-2j}R_{yy}[k]
\end{equation}
while inductive application of \eqref{eq:packetDrop} shows that
\begin{equation}
  \label{eq:packetDelta}
  R_{uu}[t] = (1-p_i)^{|t|} R_{yy}[0] + \sum_{j=1}^{|t|}\sum_{k=j}^{\infty} p_i^2 (1-p_i)^{|t|+k-2j}R_{yy}[k].
\end{equation}
Here, the sum is interpreted as $0$ when $|t|=0$.

Subtracting \eqref{eq:packetMean} from \eqref{eq:packetDelta} and
taking $Z$-transforms gives
\begin{multline}
  \theta_i(z) = \frac{(1-p_i)^2}{(1-z^{-1}(1-p_i))(1-z(1-p_i))}\cdot
  \\
  \left(R_{yy}[0]+
    \sum_{j=-\infty}^{0}\sum_{k=j}^{\infty} p_i^2 (1-p_i)^{k-2j}R_{yy}[k]
    \right).
\end{multline}

\section{Spurious Links for Perturbed LTI systems}\label{sec:LTI}
%\subsection{Network Topology Identification}\label{sec:SI}

The results reviewed from \cite{MatSal12} imply that kin relationships
could be inferred from the power spectra of ideal
measurements. However, the result of Theorem~\ref{thm:perturbStats}
implies that common types of data corruption cause perturbations to
the power spectrum of the observations. In this section, we will show
how use of the method from \cite{MatSal12} on corrupted data streams
 leads to the inference of spurious links. In
Subsection~\ref{sec:workedExample} we show how spurious links arise in
a simple example. Then in Subsection~\ref{sec:SIperturbed}, we
characterize the pattern of spurious links that could arise due to
data corruption. While these results in this section are specific to the power
spectrum inference method from \cite{MatSal12}, the work in
Section~\ref{sec:mrf} shows that the pattern of spurious links arises
more generally in network identification problems.  

%%%%%%%%%%%%%%%%%%%%%%%%%%%%%%%%%%%%%%%%%%%%%%%%%%%%%%%%%%%%%%%%%%%%%%%%%%%%%%

%{\it Remark:} We emphasize that the PSD $\Phi_{yy}$ can be computed based solely on the measurements $y_i$, $i=1,\ldots,n.$. 
%%%%%%%%%%%%%%%%%%%%%%%%%%%%%%%%%%%%%%%%%%%%%%%%%%%%%%%%%%%%%%%%%%%%%%%%%%%%%%
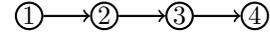
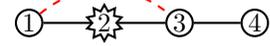
\begin{figure}[t]
  \centering
  \begin{subfigure}{0.9\columnwidth}
    \centering
        \begin{tikzpicture}[scale=0.25]
                \tikzstyle{vertex}=[circle,fill=none,minimum size=10pt,inner sep=0pt,thick,draw]
        \tikzstyle{pvertex}=[star,star points=10,fill=white,minimum size=10pt,inner sep=0pt,thick,draw]
          \node[vertex] (n1) {$1$};
          \node[vertex, right of=n1] (n2) {$2$};
          \node[vertex,right of=n2] (n3) {$3$};
          \node[vertex,right of=n3] (n4) {$4$};
          
          \draw[->,thick] (n1)--(n2);
          \draw[->,thick] (n2)--(n3);
          \draw[->,thick] (n3)--(n4);
                    
        \end{tikzpicture}
        \subcaption{\label{fig:simple Chain} Perfect Measurements}
  \end{subfigure}
    \begin{subfigure}{0.9\columnwidth}
    \centering
        \begin{tikzpicture}[scale=0.25]
                \tikzstyle{vertex}=[circle,fill=none,minimum size=10pt,inner sep=0pt,thick,draw]
        \tikzstyle{pvertex}=[star,star points=10,fill=white,minimum size=10pt,inner sep=0pt,thick,draw]
          \node[vertex] (n1) {$1$};
          \node[pvertex, right of=n1] (n2) {$2$};
          \node[vertex,right of=n2] (n3) {$3$};
          \node[vertex,right of=n3] (n4) {$4$};

          \draw[thick] (n1)--(n2);
          \draw[thick] (n2)--(n3);
          \draw[thick] (n3)--(n4);
          \draw[red,thick,dashed] (n1) to[out=40,in=140] (n3);
        \end{tikzpicture}
        \subcaption{
        \label{fig:perturb 2} Unreliable Measurements
        }
  \end{subfigure}
    \caption{
    \label{fig:corrupt spectra} %This figure shows how unreliable measurements at a node can yield in erroneous dynamic influences. 
When node $2$ has corrupt measurements an external observer might wrongly infer that the third node is directly influenced by node $1$.
  }
\end{figure}
%%%%%%%%%%%%%%%%%%%%%%%%%%%%%%%%%%%%%%%%%%%%%%%%%%%%%%%%%%%%%%%%%%%%%%%%%%%%%%
%
%Recall the  randomized delay
%described by \eqref{eq:randDelayMdl} and 
\subsection{Example: Spurious Links due to Data Corruption}
\label{sec:workedExample}
Before presenting the general results, an example will be described.
 Consider the
generative graph  of a directed chain  in Figure \ref{fig:simple
  Chain}.  Suppose the measured data-streams are denoted by
$u_i$ for node $i$ where $u_i=y_i$ for $i=1,3,4$ (thus no data
uncertainty at nodes $1,\ 3$ and $4$)  and $u_2$ is related to $y_2$
via the randomized delay model described in \eqref{eq:randDelayMdl}.
In this case, the processes $u_i$ are jointly WSS and 
the PSD of the vector process $u=\left(u_1,\cdots, u_4\right)^\top$ is
                                     related to the PSD of the vector
                                     process $y$ via: 
\begin{equation*}
\begin{split}
\Phi_{uu}(z) = & \underbrace{\begin{bmatrix}
1 & 0 & 0 &0 \\
0 & h_2(z) & 0 & 0 \\
0 & 0 & 1 & 0 \\
0 & 0 & 0 & 1
\end{bmatrix}}_{H(z)}
\Phi_{yy}(z)
\underbrace{\begin{bmatrix}
1 & 0 & 0 &0 \\
0 & h_2(z^{-1}) & 0 & 0 \\
0 & 0 & 1 & 0 \\
0 & 0 & 0 & 1
\end{bmatrix}}_{H^*(z)}\\ & \quad + 
\underbrace{\begin{bmatrix}
0 & 0 & 0 &0 \\
0 & \theta_2(z) & 0 & 0 \\
0 & 0 & 0 & 0 \\
0 & 0 & 0 & 0
\end{bmatrix}}_{D},
\end{split}
\end{equation*}
where $h_2$ and $\theta_2$ were described in
Subsection~\ref{sec:NW pert}.

Note that $\mathcal{D} = b_2 \theta _2 b_2^T$, where $b_2 = \begin{pmatrix}0 & 1 & 0
  & 0\end{pmatrix}^T$. Set $\Psi(z) = H(z) \Phi_{yy}(z)H^*(z)$. It
follows from the Woodbury matrix identity \cite{cookbook2012} that 
\begin{equation}
  \label{eq:exSparsity}
  \Phi_{uu}^{-1}(z) = \Psi^{-1}(z) - \Psi^{-1}(z) b_2 b_2^T
  \Psi^{-1}(z) \Delta^{-1},
\end{equation}
where $\Delta = \theta _2^{-1} + b_2^T\Psi^{-1}(z) b_2$ is a scalar.

Corollary~\ref{cor:informationSpectra} implies that the sparsity
pattern of $\Phi_{yy}^{-1}(z)$ is given by:
\begin{equation}
  \label{eq:origSparisty}
  \Phi_{yy}^{-1}(z) =
  \begin{bmatrix}
    * & * & 0 & 0 \\
    * & * & * & 0 \\
    0 & * & * & * \\
    0 & 0 & * & *
  \end{bmatrix}
\end{equation}
where $*$ indicates a potential non-zero entry.

Since $H(z)$ is diagonal, it follows that  $\Psi^{-1}(z)$ and $\Phi_{yy}^{-1}(z)$ have the same sparsity
pattern. Thus, the sparsity pattern of $\Psi^{-1}(z)b_2$ and 
$\Psi^{-1}(z)b_2 b_2^T \Psi^{-1}(z)$ 
are given by: 
  \begin{equation}
    \label{eq:perturbationSparsity}
\Psi^{-1}(z)b_2 = \begin{bmatrix}
  * \\ 
  * \\
  * \\
  0 
\end{bmatrix},
\:\:
 \Psi^{-1}(z) b_2 b_2^T
  \Psi^{-1}(z) = 
\begin{bmatrix}
* & * & * & 0 \\
* & * & * & 0 \\
* & * & * & 0 \\
0 & 0 & 0 & 0
\end{bmatrix} 
\end{equation}
Combining \eqref{eq:exSparsity}-\eqref{eq:perturbationSparsity}, it
follows that the $\Phi_{uu}^{-1}(z)$ has sparsity pattern given by:
\begin{equation*}
  \Phi_{uu}^{-1}(z) = \begin{bmatrix}
    * & * & \cfbox{red}{*} & 0 \\
    * & * & * & 0 \\
    \cfbox{red}{*} & * & * & * \\
    0 & 0 & * & *
    \end{bmatrix}.
  \end{equation*}

  The extra filled spot in the inverse power spectral density
  corresponds to a spurious link. See Fig.~\ref{fig:corrupt spectra}.
  
%%%%%%%%%%%%%%%%%%%%%%%%%%%%%%%%%%%%%%%%%%%%%%%%%%%%%%%%%%%%%%%%%%%%%%%%%%%%%%
\subsection{Determining Generative Topology from 
Corrupted Data Streams}\label{sec:SIperturbed} 
In this subsection, we will generalize the insights from the preceding
subsection to arbitrary DIMs. 
The following definitions are needed for the development to follow.

\begin{definition}[Path and Intermediate nodes]
Nodes $v_1,v_2,\dots ,v_k \in V$ forms a \emph{path} from $v_1$ to $v_k$ in an undirected graph $G=(V,A)$ if for every $i=1,2,\dots ,k-1$ we have $v_i- v_{i+1}$. The nodes $v_2,v_3,\dots ,v_{k-1}$  are called the \emph{intermediate nodes} in the path.
\end{definition}

\begin{definition}[Neighbors $\mathcal{N}$]
\label{def:Neighbor}
Let $G=(V,A)$ be an undirected graph. The neighbor set of node
$i$ is given by $\mathcal{N}=\{j: i-j\in A\}\cup \{i\}.$ 
\end{definition}
\begin{definition}[Erroneous Links]
\label{def:Erroneouslinks}
Let $G=(V,A)$ be an undirected graph. An edge or arc $i-j$ is
called an erroneous link when it does not belong to $A$ where
$i,j\in V$.  
\end{definition}
\begin{definition}[Perturbed Graph]
\label{def:corruptKG}
Let $G=(V,A)$ be an undirected graph. Suppose $Z\subset V$ is
the set of perturbed nodes. Then the perturbed graph of
$G$ with respect to set $Z$ is the graph
$G_Z=(V,A_Z)$ such that $i-j \in A_Z$ if either
$i-j\in A$ or there is a path from  $i$ to $j$ in $G$
such that all intermediate nodes are in $Z$. 
\end{definition}

Note that if $Z\subset \hat Z$, then $A_Z \subset A_{\hat Z}$.
%%%%%%%%%%%%%

The following theorem is the main result for LTI identification. 
 \begin{theorem}\label{thm:multiperturbation}
\it{
Consider a DIM $(\graph ,e)$ consisting of $N$ nodes with the moral graph
$G^M=(V,A^M)$. Let $Z=\{v_1,v_2,\dots,v_n\}$ be the set of
$n$ perturbed nodes where each perturbation satisfies
\eqref{eq:spectrumConditions}.  
Then 
$(\Phi_{uu}^{-1}(z))_{pq}\ne 0$ implies that $p$ and $q$ are neighbors in
the perturbed graph $G_{Z}^M$.  
}
\end{theorem}
\begin{proof}
  First, we will describe the structure of $\Phi_{uu}(z)$. For compact
  notation, we will often drop the $z$ arguments.
  
  For $p=1,\ldots,N$, if $p$ is
  not a perturbed node, set $H_p(z) =1$ and $\theta _p(z) = 0$. With this notation,
  \eqref{eq:spectrumConditions} implies that the entries of $\Phi_{uu}$ are
  given by:
  \begin{equation*}
    (\Phi_{uu})_{pq} = \begin{cases}
      H_p  (\Phi_{yy})_{pq} H_q^{*}& \textrm{if } p\ne q \\
      H_p(\Phi_{yy})_{pp} H_p^{*}+ \theta _p & \textrm{if } p = q
    \end{cases}
  \end{equation*}
  When $p\ne
  q$, there is no $\theta $ term because the perturbations were assumed to
  be independent.

  In matrix notation, we have that:
  \begin{equation*}
    \Phi_{uu} = H \Phi_{yy} H^{*} + \sum_{k=1}^n \mathcal{D}_{v_k}
  \end{equation*}
  where $H$ is the diagonal matrix with entries $H_p$ on the
  diagonal and $\mathcal{D} _{v_k}(z) = b_{v_k} \theta _{v_{k}}(z) b_{v_{k}}^T$ where
  $b_{v_{k}}$ is the canonical unit vector with $1$ at entry
  $v_{k}$. 

  Set $\Psi_0 = H \Phi_{yy} H^{*}$. For $k=0,\ldots,n-1$, we can inductively define the matrices:
  \begin{equation}
    \label{eq:inductiveMatrices}
    \Psi_{k+1} = \Psi_k + b_{v_{k+1}}\theta _{v_{k+1}} b_{v_{k+1}}^T
  \end{equation}

 For $k=1,\ldots,n$ let
$Z_k = \{v_1,\ldots,v_k\}$ and let $G_{Z_k}^M$ be the perturbed graph constructed by adding edges $i-j$ to the
original moral graph if there is a path from $i$ to $j$ whose
intermediate nodes are all in $Z_k$.  
 
We will inductively prove the following claim: For $k=1,\ldots,n$, if
$(\Psi_k^{-1})_{pq} \ne 0$, then $p$ and $q$ are neighbors in
$G_{Z_k}^M$. Proving this claim is sufficient to prove the
theorem, since $\Psi_n = \Phi_{uu}$ and $Z_n = Z$. 

First we focus on the $k=1$ case. 
Using the Woodbury Matrix identity we have,   
$\Psi ^{-1} _{1}=\Psi ^{-1}_{0}-\Gamma _{1}$, where $\Gamma_1:= (\Psi^{-1} _{0}b_{v_{1}}b^T_{v_{1}}\Psi
^{-1}_{0})\Delta ^{-1}_{v_1}$ and $\Delta _{v_{1}}= \theta_{v_{1}}^{-1} + b_{v_{1}}^T\Psi_0^{-1}(z) b_{v_{1}}$ is a scalar. Therefore, $(\Psi ^{-1} _{1})_{pq}={(\Psi ^{-1}_{0})}_{pq}-{(\Gamma _{1})}_{pq}$.

If $(\Psi ^{-1} _{1})_{pq}\neq 0$ then at least one of the conditions
(i) ${(\Psi ^{-1}_{0})}_{pq}\neq 0 $ or (ii)  ${(\Gamma _{1})}_{pq}\neq
0$ must hold. 

Suppose that $(\Psi ^{-1}_{0})_{pq}\neq 0.$ Then  $ (H^{-*}(z){{\Phi}}^{-1}_{yy}H^{-1}(z))_{pq}\neq 0.$ As $H$ is diagonal it follows that  $(\Phi ^{-1}_{yy})_{pq} \neq 0$. From  Corollary~\ref{cor:informationSpectra}, it follows that $p$ and $q$ are neighbors in $G^M$. Thus $p$ and $q$ are neighbors in $G_{B_1}^M$.

Suppose that ${(\Gamma _{1})}_{pq}\neq 0.$ Then it follows that $
{(\Psi_{0}^{-1}b_{v_{1}}b^T_{v_{1}}\Psi^{-1}_{0})}_{pq}\Delta
^{-1}_{v_{1}}\neq 0.$ Thus $ {(\Psi ^{-1}_{0}b_{v_{1}})}_{p}\neq 0$ 
and ${(b^T_{v_{1}}\Psi ^{-1}_{0})}_{q} \neq 0$.  Noting that
$\Psi_0=H\Phi_{yy} H^*$, it follows that ,
$(\Phi^{-1}_{yy})_{pv_{1}}\not =0$ and $(\Phi^{-1}_{yy})_{v_{1}q}\not
=0$. From Corollary~\ref{cor:informationSpectra} it follows that
$v_1-p$ and $v_1-q$ are edges in the moral graph
$G^M$. Thus, there is a path from $p$ to $q$ whose only
intermediate node is $v_1 \in Z_1$. 
Thus, $p,q$ are neighbors in $G_{Z_1}^M$ and the claim is verified for $k=1$. 

Now assume that the claim holds for some $k> 1$. Combining the
Woodbury matrix identity with \eqref{eq:inductiveMatrices} implies
that
\begin{equation*}
  \Psi_{k+1}^{-1} = \Psi_{k}^{-1} - \Gamma_{k+1}
\end{equation*}
where $\Gamma_{k+1} = \Psi_k^{-1} b_{v_{k+1}}b_{v_{k+1}}^T
\Psi_k^{-1}\Delta_{v_{k+1}}^{-1}$ and $\Delta _{v_{k+1}}= \theta _{v_{k+1}}^{-1} + b_{v_{k+1}}^T\Psi_k^{-1}(z) b_{v_{k+1}}$ is a scalar. 

As before, if $(\Psi_{k+1}^{-1})_{pq} \ne 0$, then either
$(\Psi_{k}^{-1})_{pq} \ne 0$ or $(\Gamma_{k+1})_{pq} \ne 0$.

If
$(\Psi_k^{-1})_{pq}\ne 0$, then the induction hypothesis implies that
$p$ and $q$ are neighbors in $G_{Z_k}^M$. Since $Z_k \subset Z_{k+1}$,
it follows that $p$ and $q$ are neighbors in $G_{Z_{k+1}}^M$.

If $(\Gamma_{k+1})_{pq} \ne 0$, then as in the $k=1$ case, we must
have that $(\Psi_{k}^{-1})_{pv_{k+1}} \ne 0$ and
$(\Psi_{k}^{-1})_{v_{k+1}q} \ne 0$. This implies that $p-v_{k+1}
\in A_{Z_k}^M$ and $v_{k+1}-q\in A_{Z_k}^M$. Thus, either $p$ and
$v_{k+1}$ are kins in the original moral graph, or there is a path from
$p$ to $v_{k+1}$ whose intermediate nodes are in $Z_k$. Similarly, for
$q$ and $v_{k+1}$. It follows that there is a path from $p$ to $q$
whose nodes are in $Z_{k+1}$, and thus $p$ and $q$ are neighbors in
$G_{Z_{k+1}}^M$. The claim, and thus the theorem, are now proved.
\end{proof}
\begin{remark}\label{rem:perturbedGraph}
Similar to Remark ~\ref{rem:otherdirectionMatSal}, cases where $i$ and $j$ are kins in the original moral graph, $G^M$, but $\Phi_{uu}^{-1}(i,j)$ is zero are pathological. $\Phi_{uu}^{-1}(i,j)$ is expressed by terms in $\Phi ^{-1}_{yy}, H_l(z)$ and $\theta _l(z)$ where $l$ is a perturbed node. As remarked earlier, the entries in $\Graph (z)$ and the corruption model described in \eqref{eq:stateSpace} must belong to a set of measure zero on space of system parameters such that $\Phi_{uu}^{-1}(i,j)$ is zero. Therefore, except for these restrictive cases, the result in Theorem ~\ref{thm:multiperturbation} implies that we can identify the perturbed kin graph. 
\end{remark}
%%%%%%%%%%%%%%%%%%%%%%%%%%%%%%%%%%%%%%%%%%%%%%%%%%%%%%%%%%%%%%%%%%%%%%%%%%%%%%%%%%%%%%%%%%%%
%\subsection{Discussion}\label{sec:disc}
%In this section, the  Wiener filtering or power spectra bases methods that were discussed in Section ~\ref{sec:LTI} are juxtaposed with the methods described in Section ~\ref{sec:nonlinear}. 
%\begin{itemize}
%\item  Suppose we consider networks with strictly causal and linear interactions. As the network size grows and the time horizon increases, the DBN structure grows exponentially. Estimating the DIR metrics gets challenging and computationally expensive. Wiener filtering based methods can be seen as a viable alternative that can be easy to compute and determine the topology of the network.
%\item IIR filters are realizable 
%\item Allows greater flexibility: dynamics with no delays and also, non-causal interactions realizable 
%\end{itemize}
%%%%%%%%%%%%%%%%%%%%%%%%%%%%%%%%%%%%%%%%%%%%%%%%%%%%%%%%%%%%%%%%%%%%%%%%%%%%%
\section{Spurious Correlations of Perturbed Markov Random Fields}
\label{sec:mrf}

So far, we have shown how perturbing time-series data can give rise to
spurious inferences. The analysis was restricted to network
identification via Wiener filtering. In this section, we will show
that spurious links arising from data corruption is a more general
phenomenon. Specifically, we will show that the exact same patterns of
spurious links from Theorem~\ref{thm:multiperturbation} will arise in a general class of probabilistic graphical models
known as Markov random fields. 

Markov random fields can model a variety of distributions, including
continuous and discrete variables. However, our presentation here is
restricted to finite-dimensional random variables with well defined
probability mass or density functions. Thus, while the class is broad,
it does not subsume the analysis from Section~\ref{sec:LTI}, which
deals with infinite-dimensional time-series data. However, as we will
see Markov random fields can model time-series analysis problems with
finite amounts of data. 

\subsection{Background on Markov Random Fields}

Our presentation of Markov random fields will be closely related to
graph cliques:
\begin{definition}[Clique]
\label{def:clique}
Given an undirected graph ${G}=(V,A)$, a \emph{clique} is a complete sub-graph formed by a set of vertices $b\subset V$ such that for all distinct $i,j \in b$ there exists $i-j \in A$. 
\end{definition}

As an example of a Markov random field, consider a finite-dimensional
version of the model from \eqref{eq:eggenmod1}:
\begin{equation}\label{eq:staticExample}
\begin{aligned} 
y_1&= e_1\\
y_2&= M_{21}y_1+e_2\\
y_3&= M_{31}y_1+e_3\\
y_4&= M_{42}y_2+M_{43}y_3+e_4\\
y_5&= M_{54}y_4+e_5
\end{aligned}
\end{equation}
with $$M=\left[\begin{array}{ccccc}
0 & 0 & 0 & 0&0\\
M_{21} & 0 & 0 & 0&0\\
M_{31}& 0 & 0 & 0 & 0\\
0 & M_{42}& M_{43}&0&0\\
0&0&0& M_{54} &0\\
\end{array}\right]
.
$$
Here, we take $e_i$ to be independent Gaussian vectors with mean $0$
and covariance $E_i$. When only a finite amount of time series data has been
collected for the system in~\eqref{eq:eggenmod1}, the relationship
between the data points can be modeled as in~\eqref{eq:staticExample}.

Now we will see how the structure of the probabilistic relationships between the
variables, $y_i$ are encoded in the corresponding moral graph from
Fig.~\ref{fig:Kin Graph}.
If $y=\begin{bmatrix}y_1 & \cdots &
  y_5\end{bmatrix}^\top$ and $e=\begin{bmatrix}e_1 & \cdots
  & e_5\end{bmatrix}^\top$, then $y = (I-M)^{-1}e$. Use the notation
$\|x\|_{E_i^{-1}}^2=x^\top E_i^{-1} x$. Then direct
calculation shows that the density of $y$ factorizes as
\begin{multline}
  p(y) = c \cdot
    \exp\left(
      -\frac{1}{2} \|y_1\|_{E_1^{-1}}^2 - \frac{1}{2}\|y_2 -
      M_{21}y_1\|_{E_2^{-1}}^2 \right. \\
    \left.
-\frac{1}{2} \|y_3-M_{31} y_1\|_{E_3^{-1}}^2\right)
    \cdot \\ \exp\left(
      -\frac{1}{2}
      \| y_4 - M_{42}y_2 - M_{43}y_3\|_{E_4^{-1}}^2
    \right) \cdot \\ 
    \cdot \exp\left((-\frac{1}{2} \|y_5 - M_{54} y_4\|_{E_5^{-1}}^2\right).
  \end{multline}

  Note that the exponential factors contain variables
  $\{y_1,y_2,y_3\}$, $\{y_2,y_3,y_4\}$, and $\{y_4,y_5\}$. These
  variable groupings correspond precisely to the maximal cliques in the moral
  graph from Fig.~\ref{fig:Kin Graph}.

  As we will discuss below, having a distribution that factorizes with
  respect to a graph is a sufficient condition for being a Markov
  random field. See also \cite{lauritzen1996gm}.
  A
  generalization of the construction of~\eqref{eq:staticExample} shows
  that finite collections of time-series data can always be viewed as
  Markov random fields. 

  To formally define Markov random fields, we need some extra notation
  and terminology.
  Let $Y$ be a collection of variables, $Y = \{y_1,\ldots,y_{|V|}\}$
  corresponding to nodes of a graph, $G=(V,A)$. If $S\subset V$, then
  we use the notation $Y_S = \{y_i | i\in S\}$.

% \begin{definition}[Active Path]
% Given an undirected graph $G$, a path between nodes $i$ and $j$ of the form $i-v_1- v_2 - \dots - v_n-j$ is \emph{active} given a set of nodes $Z$ if none of the nodes $v_1,v_2,\dots ,v_n$ belong to $Z$.
% \end{definition}
%
\begin{figure}
  \centering
          \begin{tikzpicture}[scale=0.35]
                \tikzstyle{vertex}=[circle,fill=none,minimum size=12pt,inner sep=0pt,thick,draw]
        \tikzstyle{pvertex}=[star,star points=10,fill=white,minimum size=12pt,inner sep=0pt,thick,draw]
          \node[vertex] (n1) {$1$};
          \node[pvertex, right of=n1] (n2) {$2$};
          \node[pvertex,right of=n2] (n3) {$3$};
          \node[vertex,right of=n3] (n4) {$4$};
          \node[vertex,right of=n4] (n5) {$5$};

          \node[vertex,above of=n2] (p2) {$2_p$};
          \node[vertex,above of=n3] (p3) {$3_p$};
          
          \draw[thick] (n1)--(n2);
          \draw[thick] (n2)--(n3);
          \draw[thick] (n3)--(n4);
          \draw[thick] (n4)--(n5);

          \draw[thick] (n2) -- (p2);
          \draw[thick] (n3) -- (p3);
          \draw[thick] (n1) to[out=330,in=220] (n3);
                    \draw[thick] (n2) to[out=330,in=220] (n4);
        \end{tikzpicture}
        \caption{\label{fig:cascadeJoint} Markov random field $G^J$ with
          perturbed nodes.}
      % \begin{subfigure}{0.9\columnwidth}
      %   \centering
      %     \begin{tikzpicture}[scale=0.35]
      %           \tikzstyle{vertex}=[circle,fill=none,minimum size=12pt,inner sep=0pt,thick,draw]
      %   \tikzstyle{pvertex}=[star,star points=10,fill=white,minimum size=12pt,inner sep=0pt,thick,draw]
      %     \node[vertex] (n1) {$1$};
      %     \node[vertex, right of=n1] (n2) {$2_p$};
      %     \node[vertex,right of=n2] (n3) {$3_p$};
      %     \node[vertex,right of=n3] (n4) {$4$};
      %     \node[vertex,right of=n4] (n5) {$5$};

      %     \draw[thick] (n1)--(n2);
      %     \draw[thick] (n2)--(n3);
      %     \draw[thick] (n3)--(n4);
      %     \draw[thick] (n4)--(n5);

      %   \end{tikzpicture}
      %   \caption{\label{fig:cascadeObserved} Observed Markov random field}
      % \end{subfigure}
         
\end{figure}
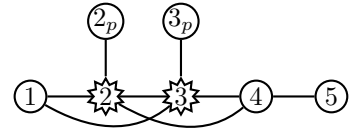
\begin{definition}[Separation]
Suppose $G=(V,A)$ is an undirected graph. Suppose, $a,b,c$ are
disjoint subsets of $V$. Then, $a$ and $b$ are \emph{separated} given
$c$ if all paths from $a$ to $b$ must pass through $c$.
\end{definition}
When $a$ and $b$ are separated given $c$, we write $\textrm{sep}(a,b\mid c)$.

  \begin{definition}[Markov random fields]
    \label{def:mrf}
    Let $Y$ be a collection of random variables associated with the
    nodes of an undirected graph, $G=(V,A)$. The variables $Y$ are
    called a \emph{Markov random
      field} with respect to $G$ if $Y_a$ and $Y_b$ are conditionally
    independent given $Y_c$ whenever $\textrm{sep}(a,b \mid c)$ holds.
\end{definition}

A useful sufficient condition for $Y$ to be a Markov random field with
respect to $G$ is for the distribution to factorize into terms
corresponding to cliques in the graph. This condition was used in the
example above. See \cite{lauritzen1996gm} for more details.

\begin{definition}[Clique Factorization]
Suppose that $Q$ is a collection of subsets of $V$ such that each $q\in Q$ forms clique in $G$. Let $P(Y)$ dentote the joint probability distribution of the random variables $Y$. Then, we say $Y$ \emph{factorizes} according to $G$, if for every $q\in Q$, there exists non-negative functions $\Psi_q $ that are functions of random variables in $q$ such that,  
\begin{equation}\label{eq:cliqueFacY}
P(Y)=\prod _{q\in Q}\Psi _q(Y_q)
\end{equation}
\end{definition}
\subsection{Inferring Erroneous Links}
\label{sec:perturbedMarkov}
Now we will describe the effects of data-corruption on inferring the
undirected graph structure from measured data.
In our work on time-series models, we assumed that individual data
streams were perturbed independently.
Here we will define a natural analog of independent perturbations for
Markov random fields. However, the perturbation models could be non-linear.

Let $Y$ be a Markov random field that factorizes with respect to a graph $G=(V,A)$. Let
$Z\subset V$ be the set of perturbed nodes. For each perturbed node,
$i\in Z$, we draw a new node $i_p$, draw an edge $i - i_p$, and denote
the
corresponding perturbed variable by $u_{i}$.
The probabilistic relationships between the original variable, $y_i$,
and the perturbed variable, $u_{i}$, is given by
$\Psi_{ii_p}(y_i,u_i)\ge 0$. Let $Z_p = \{i_p : i\in Z\}$ and let $U_{Z}$ denote
the set of perturbed variables. Then the joint distribution between
$Y$ and $U_{Z}$ can be described as:
\begin{equation}\label{eq:cliqueGt}
P(Y,U_{Z})=\prod _{q\in Q}\Psi _q(Y_q) \cdot \prod _{i\in Z}\Psi _{ii_p}(y_i,u_{i}).
\end{equation}
Since the node pairs, $\{i,i_p\}$ are
cliques, the construction above shows that the joint variables $(Y,U_Z)$ form a Markov
random field  with respect to $G^J = (V\cup Z_p,A\cup \{i - i_p : \forall
i\in Z\})$. See figure ~\ref{fig:cascadeJoint}.

Due to data corruption, only the variables $Y_{\bar Z}$ and $U_Z$ are
observed, where $\bar Z=V\setminus Z$. The next lemma shows that $(Y_{\bar{Z}},U_Z)$ is also a
Markov random field, with graph described by the perturbed graph. 
%%%%%%%%%%%%%%%%%%%%%%%%%%%%%%%%%%%%%%%%%%%%%%%%%%%%%%%%%%%%%%%%%%%%%%%%%%%%%%
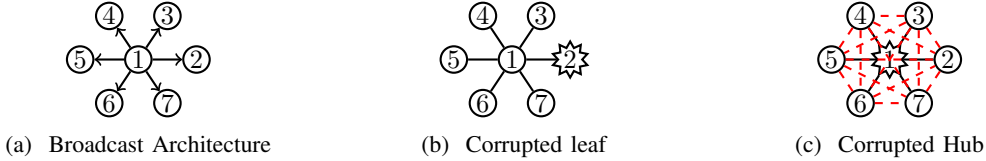
\begin{figure*}\label{fig:star}
  \centering
  \begin{subfigure}{0.55\columnwidth}
    \centering
\begin{tikzpicture}[scale=0.55]
 \tikzstyle{vertex}=[circle,fill=none,minimum size=10pt,inner sep=0pt,thick,draw]
 \tikzstyle{pvertex}=[star,star points=10,fill=white,minimum size=10pt,inner sep=0pt,thick,draw]
  \node[vertex] (n1) {$1$};
  \node[vertex] at ($(n1)+(2em,3em)$)(n3) {$3$};
  \node[vertex] (n6) at ($(n1) - (2em,3em)$) {$6$};
    \node[vertex] (n4) at ($(n1) + (-2em,3em)$) {$4$};
        \node[vertex] (n7) at ($(n1) - (-2em,3em)$) {$7$};
  \node[vertex] at ($(n1)+(4em,0)$) (n2) {$2$};
    \node[vertex] at ($(n1)-(4em,0)$) (n5) {$5$};

  \draw[->,thick] (n1)--(n2);
  \draw[->,thick] (n1)--(n3);
  \draw[->,thick] (n1)--(n4);
  \draw[->,thick] (n1)--(n5);
  \draw[->,thick] (n1)--(n6);
  \draw[->,thick] (n1)--(n7);
\end{tikzpicture}
 \subcaption{\label{fig:starDir} Broadcast Architecture}
 \end{subfigure}
%%%
  \begin{subfigure}{0.55\columnwidth}
    \centering
\begin{tikzpicture}[scale=0.55]
 \tikzstyle{vertex}=[circle,fill=none,minimum size=10pt,inner sep=0pt,thick,draw]
 \tikzstyle{pvertex}=[star,star points=10,fill=white,minimum size=10pt,inner sep=0pt,thick,draw]
  \node[vertex] (n1) {$1$};
  \node[vertex] at ($(n1)+(2em,3em)$)(n3) {$3$};
  \node[vertex] (n6) at ($(n1) - (2em,3em)$) {$6$};
    \node[vertex] (n4) at ($(n1) + (-2em,3em)$) {$4$};
        \node[vertex] (n7) at ($(n1) - (-2em,3em)$) {$7$};
  \node[pvertex] at ($(n1)+(4em,0)$) (n2) {$2$};
    \node[vertex] at ($(n1)-(4em,0)$) (n5) {$5$};

  \draw[-,thick] (n1)--(n2);
  \draw[-,thick] (n1)--(n3);
  \draw[-,thick] (n1)--(n4);
  \draw[-,thick] (n1)--(n5);
  \draw[-,thick] (n1)--(n6);
  \draw[-,thick] (n1)--(n7);
\end{tikzpicture}
 \subcaption{\label{fig:Corrupted leaf} Corrupted leaf}
 \end{subfigure}
%%%
 \begin{subfigure}{0.55\columnwidth}
    \centering
\begin{tikzpicture}[scale=0.55]
 \tikzstyle{vertex}=[circle,fill=none,minimum size=10pt,inner sep=0pt,thick,draw]
 \tikzstyle{pvertex}=[star,star points=10,fill=white,minimum size=10pt,inner sep=0pt,thick,draw]
  \node[pvertex] (n1) {$1$};
  \node[vertex] at ($(n1)+(2em,3em)$)(n3) {$3$};
  \node[vertex] (n6) at ($(n1) - (2em,3em)$) {$6$};
    \node[vertex] (n4) at ($(n1) + (-2em,3em)$) {$4$};
        \node[vertex] (n7) at ($(n1) - (-2em,3em)$) {$7$};
  \node[vertex] at ($(n1)+(4em,0)$) (n2) {$2$};
    \node[vertex] at ($(n1)-(4em,0)$) (n5) {$5$};

  \draw[-,thick] (n1)--(n2);
  \draw[-,thick] (n1)--(n3);
  \draw[-,thick] (n1)--(n4);
  \draw[-,thick] (n1)--(n5);
  \draw[-,thick] (n1)--(n6);
  \draw[-,thick] (n1)--(n7);
  \draw[red,thick,dashed] (n3)--(n4);
  \draw[red,thick,dashed] (n3)--(n5);
  \draw[red,thick,dashed] (n3)--(n6);
  \draw[red,thick,dashed] (n3)--(n7);
  \draw[red,thick,dashed] (n3)--(n2);

  \draw[red,thick,dashed] (n4)--(n5);
  \draw[red,thick,dashed] (n4)--(n6);
  \draw[red,thick,dashed] (n4)--(n7);
  \draw[red,thick,dashed] (n4)--(n2);

  \draw[red,thick,dashed] (n5)--(n6);
  \draw[red,thick,dashed] (n5)--(n7);
  \draw[red,thick,dashed] (n5)--(n2);
 
 \draw[red,thick,dashed] (n6)--(n7);
  \draw[red,thick,dashed] (n6)--(n2);
  
    \draw[red,thick,dashed] (n7)--(n2);

\end{tikzpicture}
 \subcaption{\label{fig:Corrupted Hub} Corrupted Hub}
 \end{subfigure}
 \caption{
    \label{fig:star} This figure shows an extreme example of the effect of data corruption of even a single node. \ref{fig:starDir} shows the original directed graph. \ref{fig:Corrupted leaf} shows that even if the leaf is corrupted there are no erroneous links introduced. But if the hub is corrupted as shown in \ref{fig:Corrupted Hub} then all the nodes become spuriously correlated. 
  }
 \end{figure*}
 %%%%%%%%%%%%%%%%%%%%%%%%%%%%%%%%%%%%%%%%%%%%%%%%%%%%%%%%%%%%%%%%%%%%%%%%%%%%%%
\begin{lemma}\label{lemma:UfacGz}
  Let $Y$ be a Markov random field with respect to $G=(V,A)$. Let
  $Z\subset V$ be a set of perturbed nodes and let $\bar{Z}=V\setminus
  Z$ be the unperturbed nodes. Assume that the joint distribution of
  $Y$ and the perturbed variables $U_Z$ factorizes as in
  \eqref{eq:cliqueGt}. Then the collection of observed variables $(Y_{\bar Z},U_Z)$ factorizes with respect
  to the perturbed graph $G_Z$ from Definition~\ref{def:corruptKG}.  
\end{lemma}
\begin{proof}
We will
  prove the lemma for discrete random variables. The proof for
  continuous random variables is identical except that marginalization
  would be represented by integrals instead of sums.

  Let $Z = \{v_1,\ldots,v_n\}$, $Z_0 = \emptyset$ and $Z_k =
  \{v_1,\ldots,v_k\}$.
  We will prove inductively that $(Y_{\bar Z_k},U_{Z_k})$ factorizes
  with respect to $G_{Z_k}$.

  The base case with $Z_0=\emptyset$ is immediate since $(Y_{\bar
    Z_0},U_{Z_0})=Y$ and $G_{Z_0} = G$. Now
  assume inductively that $(Y_{\bar Z_{k-1}},U_{Z_{k-1}})$ factorizes with respect
  to $G_{Z_{k-1}}$ for some $k\ge 1$.

  \begin{align}
    \MoveEqLeft
    P(Y_{\bar Z_k},U_{Z_{k}})
    \\
    \label{eq:generalMarginal}
    &= \sum_{Y_{Z_k}} \prod_{q\in Q}
                                \Psi_q(Y_q) \prod_{i=1}^k
                                \Psi_{v_i,(v_i)_p}(y_{v_i},u_{v_i})
    \\
    \nonumber
    &= \sum_{y_{v_k}} \left(
\sum_{Y_{Z_{k-1}}} \prod_{q\in Q}
                                \Psi_q(Y_q) \prod_{i=1}^{k-1}
                                \Psi_{v_i,(v_i)_p}(y_{v_i},u_{v_i})
      \right) \cdot \\
    & \hspace{20px} \Psi_{v_k,(v_k)_p}(y_{v_k},u_{v_k}) \\
    &= \sum_{y_{v_k}} P(Y_{\bar Z_{k-1}},U_{Z_{k-1}})\Psi_{v_k,(v_k)_p}(y_{v_k},u_{v_k})
  \end{align}
  The first line, \eqref{eq:generalMarginal} follows by marginalizing
  $Y_{Z_k}$ out of the factorized
  distribution from \eqref{eq:cliqueGt}. Then the terms are regrouped
  and then \eqref{eq:generalMarginal} is employed for $P(Y_{\bar
    Z_{k-1}},U_{Z_{k-1}})$.

  By induction, we have that $P(Y_{\bar
    Z_{k-1}},U_{Z_{k-1}})$ factorizes according to a collection of
  cliques, $C$,  in $G_{Z_k}$. Let $C_{v_k} \subset C$ be the
  collection of cliques such that $v_k \in c$ for all $c\in
  C_{v_k}$. For compact notation, let $X=(Y,U_Z)$. Then the formula for $P(Y_{\bar Z_k},U_{Z_k})$ can be
  expressed as
  \begin{multline}\label{eq:cliqueSplit}
    P(Y_{\bar Z_k},U_{Z_{k}}) = \left(\sum_{y_{v_k}} \prod_{c\in
        C_{v_k}} \Psi_c(X_c)\Phi_{v_k,(v_k)_p}(y_{v_k},u_{v_k})\right)
    \cdot \\
    \prod_{c\in C\setminus C_{v_k}} \Phi_c(X_c).
  \end{multline}
  The second term on the right is a collection of factors
  corresponding to cliques in $G_{Z_{k-1}}$. Now, since
  $A_{Z_{k-1}}\subset A_{Z_k}$, they must also be cliques of
  $G_{Z_k}$. The lemma will be proved if the variables first term on the right
  correspond to a clique in $G_{Z_k}$. 

  Say $i\ne v_k$ and $j\ne v_k$ are nodes corresponding to variables in the sum in
  \eqref{eq:cliqueSplit}. Then there must be paths from $i$ to $v_k$
  and $v_k$ to $j$ such that any intermediate node is in
  $Z_{k-1}$. Now, since $v_k \in Z_{k}$, there is a path from $i$ to
  $j$ such that all of the intermediate nodes are in $Z_k$. Thus, the
  nodes in the sum form a clique in $G_{Z_k}$.  
\end{proof}

In our model, we have assumed that the variables corresponding to
corrupted nodes, $y_i$ for
$i\in Z$, are hidden. Then Lemma~\ref{lemma:UfacGz} shows that
marginalizing out the variables $y_i$ introduces new probabilistic
relationships between the neighbors of $y_i$. The new links between
variables are precisely described by the perturbed graph construction
of Theorem~\ref{thm:multiperturbation}. Note that \emph{any} method that
attempts to reconstruct the graphical structure of the Markov random
field based only on the observed data that contains corrupt data will be likely to detect
spurious relationships.\footnote{In some special cases, it may be
  possible to exploit prior knowledge of network structure to rule out
  some spurious links \cite{talukdar2017exact}.}

Below, we will show
that if $P(Y_{\bar Z},U_Z)$ is positive everywhere, then the perturbed graph
exactly characterizes the conditional independence of the observed
nodes.
To present this strengthened version of Lemma~\ref{lemma:UfacGz}, some
definitions are required. 

%

%\subsection{Undirected Markov Random fields}
\begin{definition}[Pairwise Markov property]
Suppose $G=(V,A)$ is an undirected graph whose $N$ nodes represent random variables $y_1,\dots,y_N$. Let $Y=\{y_1,\dots ,y_N\}$. \emph{Pairwise Markov property} associated with $G$ holds, if for any non-adjacent vertices $i,j$, we have that 
$\textrm{sep}(i,j|V\setminus\{i,j\})$ implies that $y_i$ and $y_j$ are conditionally independent given $Y\setminus \{y_i,y_j\}$. 
\end{definition}

As in the discussion of LTI systems, it is convenient to identify the
observed but unperturbed variables $Y_{\bar Z}$ with $U_{\bar Z}$ so
that the collection of observed variables can be denoted by $U =
(U_{\bar Z},U_Z)$.

\begin{theorem}
Let $Y$ be a set of random variables that factorize according to graph
$G=(V,A).$ Suppose, $Z\subset V$, is a set of perturbed nodes such that
the joint distribution $(Y,U_Z)$ factorizes as in
\eqref{eq:cliqueGt}.
Let $U$ denote the set of all observed variables and assume that
$P(U)$ is positive everywhere.
Define $U_{\bar{i}\bar{j}}:=U\setminus \{u_i,u_j\}$.  Then, $i-j$ is
an edge in the perturbed graph, $G_Z$, if and only if $u_i$ is not conditionally independent of $u_j$ given $ U_{\bar{i}\bar{j}}$. %Then, nodes $i$ and $j$ are separated by $V\setminus \{i,j\}$ if and only if $u_i\indep u_j\mid U_{\bar{i}\bar{j}}$.Y\setminus \{y_i,y_j\}
\end{theorem}
\begin{proof}
From lemma ~\ref{lemma:UfacGz} we know that $U$ factorizes according to 
$G_Z$. Thus, positivity of $P(U)$ implies that the pairwise Markov property is equivalent to $U$ factorizing
according to $G_Z$. See \cite{lauritzen1996gm}. Therefore,
$\textrm{sep}(i,j|V\setminus\{i,j\})$ (in $G_Z$) if and only if $u_i\indep
u_j\mid U_{\bar{i}\bar{j}} $. Note that $\textrm{sep}(i,j|V\setminus\{i,j\})$ means
precisely that $i - j \notin A_Z$. 
\end{proof}
%%%%%%%%%%%%%%%%%%%%%%%%%%%%%%%%%%%%%%%%%%%%%%%%%%%%%%%%%%%%%%%%%%%%%%%%%%%%%%%%%%%%%%%%%%%
%%%%%%%%%%%%%%%%%%%%%%%%%%%%%%%%%%%%%%%%%%%%%%%%%%%%%%%%%%%%%%%%%%%%%%%%%%%%%%
\section{Simulation Results}\label{sec:results}
Power spectrum estimates were computed after  
$10^4$ simulation time steps. The estimated spectra were
then averaged over $100$ trials. The red boxes indicate the erroneous links introduced as a result of the network perturbation in addition to the the links in the true moral graph as indicated by the black boxes. % Note that PSD matrix being symmetric, we  observe the sparsity pattern in the upper triangular part.
For both the networks, the sequences $e_i$ are zero mean
white Gaussian noise.  
\subsection{Star Topology}
The transfer
function for each link is $z^{-1}$.
%We performed two experiments for the star topology from
%Figure~\ref{fig:star}: corruption of a leaf node and corruption of the
%hub. 
\subsubsection{Corrupted Leaf}
\noindent The perturbation considered here is the random delay model, \eqref{eq:randDelayMdl},
on node $2$:
\begin{equation*}
d_2[t]=\begin{cases}
3, & \textrm{ with probability } 0.65 \\
1, & \textrm{ with probability } 0.35.
\end{cases}
\end{equation*}
%
%\scriptsize
%$\\ \Phi ^{-1}_{yy}(z) =\\{\begin{bmatriy}
%16.48	& \fbox{1.50} &	\fbox{1.50}&	\fbox{1.50} &\fbox{1.49}&	\fbox{1.49}&	\fbox{1.45}\\
%1.50	&2.36	&0.06&	0.06&	0.05&	0.06	&0.05\\
%1.50	&0.06	&2.36	&0.05	&0.06	&0.06	&0.06\\
%1.50	&0.06	&0.05	&2.35	&0.06	&0.06	&0.05\\
%1.49	&0.05	&0.06	&0.06	&2.35	&0.05	&0.05\\
%1.49	&0.06	&0.06	&0.06	&0.05	&2.35	&0.05\\
%1.49	&0.05	&0.06	&0.05	&0.05	&0.05	&2.34
%\end{bmatrix}}$
\scriptsize
$\\ \Phi ^{-1}_{uu}(z) =\\{\begin{bmatrix}
15.02	& \fbox{0.14} &	\fbox{1.49}&	\fbox{1.49} &\fbox{1.50}&	\fbox{1.50}&	\fbox{1.45}\\
0.14	&1.74	&0.05&	0.05&	0.05&	0.05	&0.04\\
1.49	&0.05	&2.36	&0.05	&0.06	&0.06	&0.06\\
1.49	&0.05	&0.05	&2.35	&0.06	&0.05	&0.06\\
1.50	&0.05	&0.06	&0.06	&2.36	&0.05	&0.05\\
1.50	&0.05	&0.06	&0.05	&0.05	&2.36	&0.05\\
1.45	&0.04	&0.06	&0.06	&0.05	&0.05	&2.34
\end{bmatrix}}$
\normalsize

As predicted by Theorem~\ref{thm:multiperturbation}, perturbation of Node $2$ for this
architecture does not introduce any erroneous links.  See Figure~\ref{fig:Corrupted leaf}.

\subsubsection{Corrupted Hub}
\noindent The perturbation considered here is a random delay on the
hub node:
\begin{equation*}
d_1[t]=\begin{cases}
2, & \textrm{ with probability } 0.75 \\
4, & \textrm{ with probability } 0.25.
\end{cases}
\end{equation*}

Theorem~\ref{thm:multiperturbation} predicts that perturbing the
central node could introduce erroneous links between all of the
nodes. See Figure~\ref{fig:Corrupted Hub}.

%$\\ \Phi ^{-1}_{yx}(z) =\\{\begin{bmatrix}
%16.41	& \fbox{1.46} &	\fbox{1.50}&	\fbox{1.49} &\fbox{1.47}&	\fbox{1.49}&	\fbox{1.50}\\
%1.46	&2.35	&0.05&	0.06&	0.06&	0.06	&0.06\\
%1.50	&0.05	&2.35	&0.06	&0.06	&0.06	&0.06\\
%1.49	&0.06	&0.06	&2.34	&0.05	&0.06	&0.05\\
%1.47	&0.06	&0.06	&0.05	&2.34	&0.05	&0.05\\
%1.49	&0.06	&0.06	&0.06	&0.05	&2.35	&0.06\\
%1.50	&0.06	&0.06	&0.05	&0.05	&0.06	&2.36
%\end{bmatrix}}$
\scriptsize
$\\ \Phi ^{-1}_{uu}(z) =\\{\begin{bmatrix}
5.08	& \fbox{0.40} &	\fbox{0.40}&	\fbox{0.40} &\fbox{0.39}&	\fbox{0.39}&	\fbox{0.38}\\
0.40	&2.07	&\cfbox{red}{0.27}&	\cfbox{red}{0.27}&	\cfbox{red}{0.27}&	\cfbox{red}{0.26}	&\cfbox{red}{0.27}\\
0.40	&0.27	&2.08	&\cfbox{red}{0.27}	&\cfbox{red}{0.27}	&\cfbox{red}{0.28}	&\cfbox{red}{0.27}\\
0.40	&0.27	&0.27	&2.07	&\cfbox{red}{0.27}	&\cfbox{red}{0.27}	&\cfbox{red}{0.27}\\
0.39	&0.27	&0.27	&0.27	&2.07	&\cfbox{red}{0.27}	&\cfbox{red}{0.27}\\
0.39	&0.26	&0.28	&0.27	&0.27	&2.08	&\cfbox{red}{0.27}\\
0.38	&0.27	&0.27	&0.27	&0.27	&0.27	&2.08
\end{bmatrix}}$
\normalsize

%The estimated power spectrum indicates that indeed, spurious links may
%be predicted from examination of the power spectrum $\Phi_{uu}$. 
%%
%Therefore, to obtain accurate topology inferences, high fidelity
%measurements at central nodes would be beneficial. 
\subsection{Chain Topology}
The chain topology  in
Figure~\ref{fig:cascade} is considered. The transfer functions are: between nodes $1$ and $2$,  $1.2+0.9z^{-1}$, between nodes $2$ and $3$,  $1+0.2z^{-1}$, between nodes $3$ and $4$,  $1-0.9z^{-1}+0.3z^{-2}$ and then for the last link $z^{-1}$. 
Figure~\ref{fig:cascade}
In the simulations, nodes $2$ and $3$ are simultaneously corrupted
with the random delay models 
\begin{equation*}
d_2[t]=\begin{cases}
1, & \textrm{ with probability } 0.83 \\
2, & \textrm{ with probability } 0.17.
\end{cases}
\end{equation*}
\begin{equation*}
d_3[t]=\begin{cases}
2, & \textrm{ with probability } 0.85 \\
4, & \textrm{ with probability } 0.15.
\end{cases}
\end{equation*}
%$\\ \Phi ^{-1}_{yx}(z) =\\{\begin{bmatrix}
%5.77	& \fbox{2.29} &0.07&	0.07 &0.05\\
%2.29	&2.91	&\fbox{1.44}&	0.06&	0.04\\
%0.07	&1.44	&4.14	&\fbox{1.63}	&0.04\\
%0.07 	&0.06	&1.63	&2.84	&\fbox{0.90}\\
%0.05	&0.04	&0.04	&0.90	&1.42
%\end{bmatrix}}$
\scriptsize
$\\ \Phi ^{-1}_{uu}(z) =\\{\begin{bmatrix}
4.23	& \fbox{0.54} &	\cfbox{red}{0.12}&	\cfbox{red}{0.25} & 0.05\\
0.54	&1.20	&\fbox{0.16}&	\cfbox{red}{0.13}&	0.02\\
0.12	&0.16	&1.06	&\fbox{0.12}	&0.02\\
0.25	&0.13	&0.12	&2.22	&\fbox{0.90}\\
0.05	&0.02	&0.02	&0.90	&1.42
\end{bmatrix}}$
\normalsize

Perturbation of $2$ adds a false relationship between $1$ and $3$. In addition, perturbation of $3$ introduces erroneous relations between the nodes $1$ and $4$ as well as between $2$ and $4$. Thus the erroneous relationships could arise between any nodes that are kins of $3$ including the already introduced false kins of $3$. Despite this cascaded effect the erroneous links remain local in the sense that the dependency of $5$ is unaffected. 
%%%%%%%%%%%%%%%%%%%%%%%%%%%%%%%%%%%%%%%%%%%%%%%%%%%%%%%%%%%%%%%%%%%%%%%%%%%%%%
\begin{figure}
  \centering
  \begin{subfigure}{0.9\columnwidth}
    \centering
        \begin{tikzpicture}[scale=0.25]
                \tikzstyle{vertex}=[circle,fill=none,minimum size=10pt,inner sep=0pt,thick,draw]
        \tikzstyle{pvertex}=[star,star points=10,fill=white,minimum size=10pt,inner sep=0pt,thick,draw]
          \node[vertex] (n1) {$1$};
          \node[pvertex, right of=n1] (n2) {$2$};
          \node[vertex,right of=n2] (n3) {$3$};
          \node[vertex,right of=n3] (n4) {$4$};
          \node[vertex,right of=n4] (n5) {$5$};

          \draw[thick] (n1)--(n2);
          \draw[thick] (n2)--(n3);
          \draw[thick] (n3)--(n4);
          \draw[thick] (n4)--(n5);
          \draw[red,thick,dashed] (n1) to[out=40,in=140] (n3);
          
        \end{tikzpicture}
        \subcaption{\label{fig:cascade2} Node 2 Perturbed}
  \end{subfigure}
    \begin{subfigure}{0.9\columnwidth}
    \centering
        \begin{tikzpicture}[scale=0.25]
                \tikzstyle{vertex}=[circle,fill=none,minimum size=10pt,inner sep=0pt,thick,draw]
        \tikzstyle{pvertex}=[star,star points=10,fill=white,minimum size=10pt,inner sep=0pt,thick,draw]
          \node[vertex] (n1) {$1$};
          \node[vertex, right of=n1] (n2) {$2$};
          \node[pvertex,right of=n2] (n3) {$3$};
          \node[vertex,right of=n3] (n4) {$4$};
          \node[vertex,right of=n4] (n5) {$5$};

          \draw[thick] (n1)--(n2);
          \draw[thick] (n2)--(n3);
          \draw[thick] (n3)--(n4);
          \draw[thick] (n4)--(n5);
          \draw[red,thick,dashed] (n2) to[out=40,in=140] (n4);
        \end{tikzpicture}
        \subcaption{
        \label{fig:cascade3} Node 3 Perturbed
        }
  \end{subfigure}
  
      \begin{subfigure}{0.9\columnwidth}
    \centering
        \begin{tikzpicture}[scale=0.35]
                \tikzstyle{vertex}=[circle,fill=none,minimum size=12pt,inner sep=0pt,thick,draw]
        \tikzstyle{pvertex}=[star,star points=10,fill=white,minimum size=12pt,inner sep=0pt,thick,draw]
          \node[vertex] (n1) {$1$};
          \node[pvertex, right of=n1] (n2) {$2$};
          \node[pvertex,right of=n2] (n3) {$3$};
          \node[vertex,right of=n3] (n4) {$4$};
          \node[vertex,right of=n4] (n5) {$5$};k

          \draw[thick] (n1)--(n2);
          \draw[thick] (n2)--(n3);
          \draw[thick] (n3)--(n4);
          \draw[thick] (n4)--(n5);
          \draw[red,thick,dashed] (n1) to[out=40,in=140] (n3);
          \draw[red,thick,dashed] (n2) to[out=40,in=140] (n4);
          \draw[red,thick,dashed] (n1) to[out=-35,in=215] (n4);
        \end{tikzpicture}
        \subcaption{
        \label{fig:cascade23}
        Nodes 2 and 3 Perturbed
        }
  \end{subfigure}
  \caption{
    \label{fig:cascade} This figure shows how multiple perturbations
    can lead to a cascade effect  as predicted by Theorem \ref{thm:multiperturbation}.
     Here the original moral graph is a chain. \ref{fig:cascade2} and \ref{fig:cascade3} show the erroneous
    edges that can arise from perturbing a single node. If nodes $2$ and $3$ are both
    perturbed, then another erroneous link between $1$ and $4$ must be
    added.
  }
\end{figure}
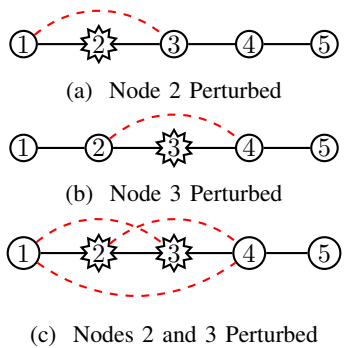
%%%%%%%%%%%%%%%%%%%%%%%%%%%%%%%%%%%%%%%%%%%%%%%%%%%%%%%%%%%%%%%%%%%%%%%%%%%%%%%
\section{Conclusion}\label{sec:conclude}
We studied the problem of inferring the network structure of
interacting agents from corrupt data-streams. We described general
model of data-corruption that introduces an additive term in the power
spectra and captures a wide class of measurement uncertainties. We
then studied inferring topology of a network of LTI systems from
corrupt data-streams. We established that network topology
reconstruction from corrupt data streams can result in erroneous links
between the nodes. Particularly we provided exact characterization by
proving that  the erroneous links are localized to the neighborhood of
the perturbed node. We then studied the influence of data corruption
on Markov random field models. Here we found that our characterization
of erroneous links for LTI systems precisely characterized the
spurious relationships that can arise in Markov random fields.

Our results show that data corruption gives
rise to the appearance of cliques that are localized around the corrupt
nodes.
Two natural future research directions emerge. The first direction would be to
prior structural knowledge to infer the location of corrupt nodes. For
example, in some power network problems, cliques cannot be present,
and so the appearance of a clique would indicate that data must have
been corrupted.
The other direction would be to use network reconstruction results of to
guide sensor placement algorithms. For example, if the neighborhood of
a node forms a clique, then our results suggest that this clique may
be due to data corruption, and thus a better sensor could be used to
rule out this possibility. 
\appendices

\section{Proof of Theorem~\ref{thm:perturbStats}}
\label{app:perturbationProof}

Define the following deviations from the mean: $\Delta A_i[t] = A_i[t] - \bar A_i$,
$\Delta B_i[t] = B_i[t] - \bar B_i$,
$\Delta C_i[t] = C_i[t] - \bar C_i$,
and $\Delta D_i[t] = D_i[t] - \bar D_i$

Note that the Lyapunov equation, \eqref{eq:genLyap}, can be expressed
as:
\begin{equation}
  P = \bar A_i^\top P \bar A_i + \mean[\Delta A_i[t]^\top P \Delta
  A_i[t]] + Q \succeq \bar A_i^\top P \bar A_i + Q.
\end{equation}
Here $S \preceq T$ denotes that $T - S$ is positive semidefinite. 
Since a solution must hold for all $Q$, it must hold, in particular
for positive definite $Q$. Thus, $\bar A_i$ must be a stable matrix. 

Set 
$\bar u_i[t]=(h_i \star y_i)[t] = \mean[u_i[t]| y_i]$, so that $\Delta
u_i[t] = u_i[t]-\bar u_i[t]$. 

With this notation, the cross spectrum, \eqref{eq:crossSpectrum}, will be derived:
\begin{align}
  R_{u_i y_i}[t] &= \mean[u_i[t]y_i[0]] \\
  \label{eq:crossTower}
                 &= \mean[ \mean[u_i[t]y_i[0]|y_i]] \\
                 &= \mean[ (h_i \star y_i)[t] y_i[0]] \\
  &= (h_i \star R_{y_iy_i})[t].
\end{align}
Here, \eqref{eq:crossTower} is due to the tower property of
conditional expectation. Then \eqref{eq:crossSpectrum} follows by taking
$Z$-transforms.

Since $\bar A_i$ is stable and $y_i[t]$ is wide-sense stationary, we
must have that $\bar u_i[t]$ is wide-sense stationary.

Note that by construction, $R_{u_iu_i}[t]= R_{\bar
  u_i \bar u_i}[t] + R_{\Delta u_i \Delta u_i}[t]$. Furthermore, we
must have that
\begin{equation}
  R_{\bar
  u_i \bar u_i}[t] = (h_i \star R_{yy} \star h_i^*)[t],
\end{equation}
where $h_i^*$ is the time-reversed, transposed impulse response. Thus,
\eqref{eq:uSpectrum} holds by taking $Z$-transforms.

The only part that remains to be proved is that $u_i$ is wide-sense
stationary. This will follow as long as $\Delta u_i[t]$ has a finite
autocorrelation.

To show that $R_{\Delta u_i \Delta u_i}[t]$ is bounded, we will explicitly construct
an expression for it. To derive this expression, we need expressions
for the autocorrelation of $x_i$ and the cross correlation between
$x_i$ and $y_i$.

Let $\bar x_i[t] =      \left(\left[ \begin{array}{c|c}
         \bar A_i & \bar B_i \\
                                       \hline
                                       I & 0
       \end{array}
     \right] \star y_i \right)[t]$
    and let $\Delta x_i[t] = x_i[t]-\bar x_i[t]$. Note that $\bar
    x_i[t] = \mean[x_i[t]|y_i]$.  As with $\bar u_i$, we have that
    $\bar x_i[t]$ is wide-sense stationary.
Using a derivation identical to that of
$R_{u_iy_i}[t]$, we have that the cross correlation of $x_i$ and $y_i$
is given by:
\begin{equation}
    R_{x_iy_i}[t] = \left(\left[
      \begin{array}{c|c}
        \bar A_i & \bar B_i \\
        \hline
        I & 0
      \end{array}
      \right] \star R_{y_iy_i}\right)[t]
  \end{equation}
  Thus, we see that $R_{x_iy_i}[t] = R_{\bar x_i y_i}[t]$.

    Now we will work out the autocorrelation of $x_i$. 
    The
    autocorrelation of $\bar x_i[t]$ is given by:
       \begin{equation}
        R_{\bar x_i\bar x_i}[t] = \left(\left[
      \begin{array}{c|c}
        \bar A_i & \bar B_i \\
        \hline
        I & 0
      \end{array}
      \right] \star R_{y_iy_i} \star \left[
      \begin{array}{c|c}
        \bar A_i & \bar B_i \\
        \hline
        I & 0
      \end{array}
      \right]^*\right)[t].
  \end{equation}
   
 By construction, we have that $R_{x_i x_i}[t] = R_{\bar x_i \bar x_i}[t] + R_{\Delta x_i \Delta x_i}[t]$. 
 The following lemma characterizes the
 autocorrelations of $\Delta x_i[k]$.

    \begin{lemma}
      Assume that a solution to the generalized Lyapunov equation,
      \eqref{eq:genLyap}, holds for all $Q$. 
      Then $R_{\Delta x_i \Delta x_i}[0]$ is uniquely defined by:
      \begin{multline}
        R_{\Delta x_i\Delta x_i}[0]= 
\mean[A_i[0] R_{\Delta x_i\Delta x_i}[0] A_i[0]^\top] + W
        \\  +\mean\left[
          \begin{bmatrix}
                            \Delta A_i[0] & \Delta B_i[0] 
                          \end{bmatrix}
                          \begin{bmatrix}
                            R_{\bar x_i\bar x_i}[0] &
                            R_{\bar x_iy_i}[0] \\
                            R_{y_i\bar x_i}[0] & R_{y_iy_i}[0]
                          \end{bmatrix}
                          \begin{bmatrix}
                            \Delta A_i[0]^\top \\ \Delta B_i[0]^\top 
                          \end{bmatrix}
        \right].
      \end{multline}
      For $k>0$,
      \begin{align*}
        R_{\Delta x_i \Delta x_i}[k] & = \bar A_i^k R_{\Delta x_i \Delta x_i}[0]
        \\
        R_{\Delta x_i \Delta x_i}[-k] &= R_{\Delta x_i \Delta x_i}[k]^\top.
      \end{align*}
      \end{lemma}

      \begin{proof}
        For $k>0$ we  have
        \begin{align*}
          \MoveEqLeft
          R_{\bar x_i\bar x_i}[k] + R_{\Delta x_i \Delta x_i}[k]
          \\
          & = \mean[x_i[k] x_i[0]^\top] \\
          &= \mean[(A_i[{k-1}]x_i[{k-1}]+B_i[{k-1}]
                                y_i[{k-1}])x_i[0]^\top] \\
                              &=\bar A_i R_{x_ix_i}(k-1)+\bar B_i R_{y_ix_i}(k-1) \\
          &= (\bar A_i R_{\bar x_i\bar x_i}(k-1) + \bar B_i R_{y_ix_i}(k-1))\\& +
            \bar A_i R_{\Delta x_i \Delta x_i}(k-1) \\
          &= R_{\bar x_i \bar x_i}[k]+ \bar A_i R_{\Delta x_i \Delta x_i}[k-1].
        \end{align*}
        Thus, the formula for $R_{\Delta x_i \Delta x_i}[k]$ holds for
        $k\ne 0$. (The expression for $k<0$ follows from transposing.)

        Note that
        \begin{align*}
          \MoveEqLeft
          \Delta x_i[{k+1}]
          \\
          &= (\bar A_i + \Delta A_i[k])(\bar x_i[k]+\Delta x_i[k])
            +(\bar B_i + \Delta B_i[k]) y_i[k] \\&
                                                  + w_i[k]- \bar A_i \bar x_i[k] - \bar
                           B_i y_i[k] \\
          &= A_i[k] \Delta x_i[k] + \Delta A_i[k] \bar x_i[k] + \Delta
            B_i[k] y_i[k] + w_i[k].
        \end{align*}
        Furthermore, note that $\Delta x_i[k]$ is independent of $\Delta
        A_i[k]$ and $\Delta B_i[k]$. The expression for $R_{\Delta x_i \Delta
          x_i}(0)$ follows by setting $\mean[\Delta x_i[{k+1}]\Delta
        x_i[{k+1}]^\top]=\mean[\Delta x_i[k]\Delta x_i[k]^\top]$.
        
        Note that $R_{\Delta x_i \Delta x_i}(0)$ can be computed from \eqref{eq:genLyap} with
        \begin{multline}
          Q = W + \\ \mean\left[
          \begin{bmatrix}
                            \Delta A_i[0] & \Delta B_i[0]
                          \end{bmatrix}
                          \begin{bmatrix}
                             R_{\bar x_i\bar x_i}(0) &
                            R_{\bar x_iy_i}(0) \\
                            R_{y_i\bar x_i}(0) & R_{y_iy_i}(0)
                          \end{bmatrix}
                          \begin{bmatrix}
                            \Delta A_i[0]^\top \\ \Delta B_i[0]^\top 
                          \end{bmatrix}
        \right] \end{multline}
      \end{proof}

      As discussed above, the proof of the theorem will be completed once the
      autocorrelation of $\Delta u_i$ is characterized. The following
      lemma gives the desired characterization.
    
        \begin{lemma}
          \label{lem:deltaU}
          For $k=0$, $R_{\Delta u_i\Delta u_i}[0]$ is given by
          
          \begin{multline}
            R_{\Delta u_i \Delta u_i}[0] = \bar C_i R_{\Delta x_i \Delta
              x_i}[0]\bar C_i^\top + V \\  + \mean
            \left[
                     \begin{bmatrix}
                            \Delta C_i[0] & \Delta D_i[0] 
                          \end{bmatrix}
                          \begin{bmatrix}
                            R_{x_ix_i}[0] &
                            R_{x_iy_i}[0] \\
                            R_{y_ix_i}[0] & R_{y_iy_i}[0]
                          \end{bmatrix}
                          \begin{bmatrix}
                            \Delta C_i[0]^\top \\ \Delta D_i[0]^\top 
                          \end{bmatrix}
            \right]
          \end{multline}
          For $k> 0$, $R_{\Delta u_i\Delta u_i}[k]$ is given by
          \begin{multline}
            R_{\Delta u_i \Delta u_i}[k] = \bar C_i R_{\Delta x_i
              \Delta x_i}[k]
            \bar C_i^\top  + \bar C_i \bar A_i^{k-1} S \\ +\bar C_i \bar A_i^{k-1}\mean
            \left[
                     \begin{bmatrix}
                            \Delta A_i[0] & \Delta B_i[0] 
                          \end{bmatrix}
                          \begin{bmatrix}
                            R_{x_ix_i}[0] &
                            R_{x_iy_i}[0] \\
                            R_{y_ix_i}[0] & R_{y_iy_i}[0]
                          \end{bmatrix} \cdot
                          \right. \\
                          \left.
                          \begin{bmatrix}
                            \Delta C_i[0]^\top \\ \Delta D_i[0]^\top 
                          \end{bmatrix}
            \right]
          \end{multline}
          For $k<0$,  $R_{\Delta u_i \Delta u_i}[k] = R_{\Delta u_i \Delta
            u_i}[-k]$. 
\end{lemma}
\begin{proof}
 Note that $\Delta u_i[k]$ can be decomposed as:
  \begin{align}
    \MoveEqLeft[1]
    \Delta u_i[k] \\
    &= u_i[k]- \bar u_i[k]  \\
    &= (\bar C_i + \Delta C_i[k]) (\bar x_i[k] + \Delta x_i[k]) + (\bar D_i +
      \Delta D_i[k]) y_i[k]\\& + v_i[k]-\bar C_i\bar x_i[k]-\bar D_iy_i[k]\\
    \label{eq:deltaU}
    &= \bar C_i \Delta x_i[k] + \Delta  C_i[k] x_i[k] + \Delta D_i[k]
      y_i[k] + v_i[k]
  \end{align}
 
  As before, $\Delta x_i[k]$ is independent of $\Delta C_i[k]$ and $\Delta
  D_i[k]$. Thus, the expression for $R_{\Delta u_i\Delta u_i}[0]$ follows by
  computing $\mean[\Delta u_i[k]^2]$.

  For $k>0$, note that $\Delta C_i[k]$ and $\Delta D_i[k]$ are independent
  of $\Delta C_i[0]$ and $\Delta D_i[0]$. However, $\Delta x_i[k]$ may be
  correlated with $\Delta C_i[0]$, $\Delta D_i[0]$, and $v_i[0]$. So, multiplying
  the expression from \eqref{eq:deltaU} for $k>0$ and $k=0$ and
  dropping the $\Delta C_i[k]$ and $\Delta D_i[k]$ terms gives
  \begin{align}
    \nonumber
    \MoveEqLeft
    R_{\Delta u_i \Delta u_i}(k) 
    =
                               \mean\left[
                               \bar C_i \Delta x_i[k]
                               (\bar C_i\Delta x_i[0] + )^\top 
      \right] \\
    & + \mean\left[\bar C_i \Delta x_i[k] \left(\Delta
                               C_i[0] x_i[0] + \Delta D_i[0] y_i[0]+v_i[0]\right)^\top\right]
      \\
    \label{eq:deltaUCoupled}
    &= \bar C_i R_{\Delta x_i \Delta x_i}(k) \bar C_i^\top \\ & +
    \bar C_i \mean [\Delta x_i[k] (\Delta C_i[0] x_i[0] + \Delta
                                                                D_i[0]
                                                                y_i[0]
                                                                + v_i[0])^\top]
  \end{align}

  Let $A_i[j:k]$ be the product defined by $A_{i}[k:k] = I$ and $A_{i}[j:k] =
  A_{i,j}[{k-1}]A_i[{k-2}]\cdots A_i[j]$ for $j<k$. An induction argument shows
  that 
  \begin{align*}
      x_i[k] &= A_i[0:k] x_i[0] +  \sum_{j=0}^{k-1} A_i[{j+1}:{k}]( B_i[j] y_i[j]+w_i[j]) \\
           &= A_i[1:k] A_i[0] x_i[0] + B_i[0] y_i[0] + w_i[0]) + \\
           & + \sum_{j=1}^{k-1}
      A_i[{j+1}:{k}](B_i[j] y_i[j] + w_i[j]).
  \end{align*}
  
  Let $\mathcal{F}$ be the $\sigma$-algebra generated by $y_i$ and all
  of the random terms $(A_i[j],B_i[j],C_i[j],D_i[j],w_i[j],v_i[j])$
  for $i\le 0$. Then the
  expression for $x_i[k]$ implies that
  \begin{align*}
  \MoveEqLeft[2]
    \mean[x_i[k]|\mathcal{F}] =  \sum_{j=1}^{k-1}\bar A^{k-1-j} \bar B y_i[j] \\
    & +
                                \bar A^{k-1}\left(
      (\bar A +\Delta A_i[0])x_i[0] + (\bar B +\Delta B_0)y_i[0]
                             +w_i[0]\right) \\ 
                              &
                                = \bar x_i[k]+ \bar A^{k-1}\bar A
                                \Delta x_i[0] + \\
   &+ \bar A^{k-1}(\Delta A_i[0] x_i[0] + \Delta B_i[0] y_i[0]+w_i[0]).
  \end{align*}
  Using the tower property gives:
  \begin{align*}
    \MoveEqLeft
    \mean [\Delta x_i[k] (\Delta C_i[0] x_i[0] + \Delta D_i[0]
    y_i[0]+v_i[0])^\top]
    \\
    &=
\mean[  \mean [\Delta x_i[k] (\Delta C_i[0] x_i[0]  + \Delta D_i[0] y_i[0]+v_i[0])^\top
                                                                |\mathcal{F}]]
    \\
    &= \bar A^{k-1} \mean[(\Delta A_i[0] x_i[0] + \Delta B_i[0]
      y_i[0]+w_i[0])\\& \qquad \cdot (\Delta C_i[0]
                 x_i[0] + \Delta D_i[0] y_i[0]+v_i[0])^\top],
  \end{align*}
  where the last equality used that $\Delta x_i[0]$ is independent of
  $\Delta A_i[0]$, $\Delta B_i[0]$, and $v_i[0]$. 
Combining this result with
  \eqref{eq:deltaUCoupled} gives the desired expression for $R_{\Delta
    u\Delta u}(k)$. The expression for $R_{\Delta u\Delta u}(-k)$
  follows because $\Delta u_i[k]$ is a real scalar. 
\end{proof}

\bibliographystyle{IEEEtran}
\bibliography{ref}   
\begin{IEEEbiography} [{\includegraphics[width=1in,height=1.25in,clip,keepaspectratio]{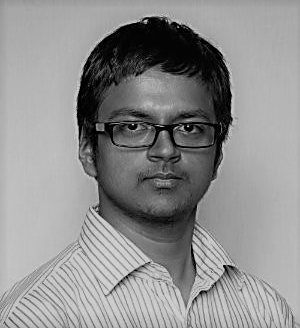}}]{Venkat Ram Subramanian}
received the B.Tech degree in electrical engineering from SRM University, Chennai, India, in 2014, and the M.S. degree in electrical engineering from the University of Minnesota, Minneapolis, in 2016. 

Currently, he is working towards a Ph.D. degree at the University of Minnesota. His Ph.D. research is on learning dynamic relations in networks from corrupt data-streams. In addition to system identification and stochastic systems, his research interests also include optimal control and graphical models.  
\end{IEEEbiography}

\begin{IEEEbiography}[{\includegraphics[width=1in,height=1.25in,clip,keepaspectratio]{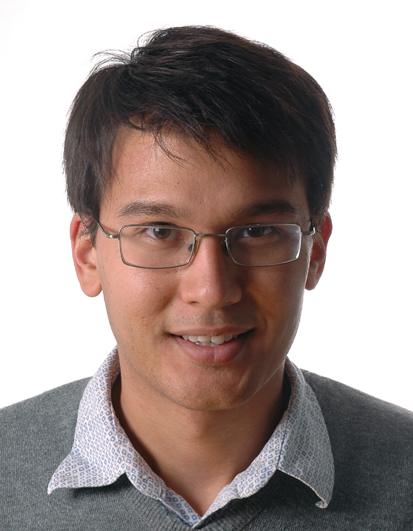}}] {Andrew Lamperski}
(S'05--M'11) received the B.S. degree in biomedical engineering and
mathematics in 2004 from the Johns Hopkins University, Baltimore, MD,
and the Ph.D. degree in control and dynamical systems in 2011 from the
California Institute of Technology, Pasadena. He held postdoctoral
positions in control and dynamical systems at the California Institute
of Technology from 2011--2012 and in mechanical engineering at The
Johns Hopkins University in 2012. From 2012--2014,
did
postdoctoral work in the Department of Engineering, University of
Cambridge, on a scholarship from the Whitaker International
Program. In 2014, he joined the Department of Electrical and Computer
Engineering, University of Minnesota as an Assistant Professor. His
research interests include optimal control, optimization, and identification, with applications to neuroscience and robotics.
\end{IEEEbiography}

\begin{IEEEbiography}[{\includegraphics[width=1in,height=1.25in,clip,keepaspectratio]{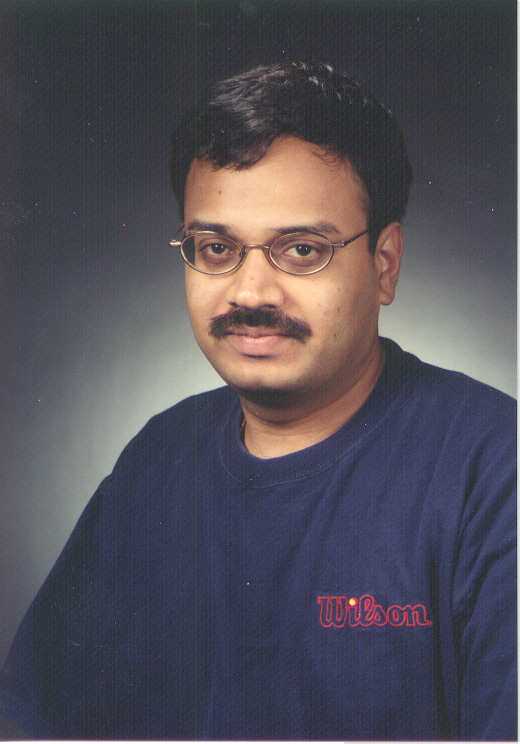}}] {Murti Salapaka} (SM'01--F'19)
  Murti Salapaka received the bachelor’s degree from the Indian 
Institute of Technology, Madras, India, in 1991, and the Master’s and 
Ph.D. degrees from the University of California, Santa Barbara, CA, USA, 
in 1993 and 1997, respectively, all in mechanical engineering. He was 
with Electrical Engineering department, Iowa State University, from 1997 
to 2007. He is currently the Vincentine Hermes-Luh Chair Professor with 
the Electrical and Computer Engineering Department, University of 
Minnesota, Minneapolis, MN, USA. Prof. Salapaka was the recipient of the 
NSF CAREER Award and the ISU—Young Engineering Faculty Research Award 
for the years 1998 and 2001, respectively. He is an IEEE Fellow.
\end{IEEEbiography}

\end{document}